\documentclass[11pt,fleqn]{article}

\usepackage{amssymb}
\usepackage{amsthm}
\usepackage{graphicx}
\usepackage{epstopdf}
\usepackage{titlesec}
\usepackage{natbib}

\titlelabel{\thetitle.\quad}
\newtheorem{theorem}{Theorem}

\newtheorem{lemma}[theorem]{Lemma}

\theoremstyle{plain}
\newtheorem{Def}{Definition}

\def\qed{$\Box$ \smallskip}
\bibpunct{(}{)}{;}{a}{,}{;}

\begin{document}
\title{Near-Optimal Disjoint-Path Facility Location Through\\ Set Cover by Pairs
\footnote{Copyright 2016 AT\&T Intellectual Property. All rights reserved.}}

\author{
	\sc David S. Johnson
	\thanks {Died March 8, 2016.
	This work was performed when the author was an employee of AT\&T Labs-Research.
	}
	\and
	\sc Lee Breslau
	\thanks {AT\&T Labs-Research, 1 AT\&T Way,
        Bedminster, NJ 07921.
	}
	\and
	\sc Ilias Diakonikolas
	\thanks{School of Informatics, University of Edinburgh,
	10 Crichton Street, Edinburgh EH8 9AB, Scotland.  This work was
	performed when the author was a student at Columbia University,
	and was partially supported by AT\&T.
	}
	\and
	\sc Nick Duffield
	\thanks {Department of Electrical and Computer Engineering, Texas A\&M
	University, College Station, TX 77843.
	This work was performed when the author was an employee of AT\&T Labs-Research.
	}
	\and
	\sc Yu Gu
	\thanks{Amazon Web Services, 1200 12th Avenue South, Seattle, WA 98144.
	This work was
	performed when the author was a student at the University of Massachusetts,
	and was partially supported by AT\&T.
	}
	\and
	\sc MohammadTaghi Hajiaghayi
	\thanks {Computer Science Department, University of Maryland,
        A.V. Williams Bldg., Room 3249,
        College Park, MD 20742.
	This work was performed when the author was an employee of AT\&T Labs-Research.
	}
	\and
	\sc Howard Karloff
	\thanks {
	This work was performed when the author was an employee of AT\&T Labs-Research.
	}
	\and
	\sc Mauricio G. C. Resende
	\thanks {Amazon.com, 300 Boren Avenue North, Seattle, WA 98109.
	This work was performed when the author was an employee of AT\&T Labs-Research.
	}
	\and
	\sc Subhabrata Sen$\ ^\ddag$
}

\date{\today}
\maketitle

\sloppy
\def\opt{\textsc{Opt}}
\def\greedy{\textsc{Greedy}}
\def\gr4{\textsc{Greedy(400)}}
\def\genetic{\textsc{Genetic}}
\def\dhs{\textsc{DHS}}
\def\hh{\textsc{DHS}}
\def\shs{\textsc{SHS}}
\def\mip{\textsc{MIP}}
\def\lb{\textsc{HSLB}}
\def\dsj{\textsc{DP}}

\vspace{-.4in}
\begin{abstract}

In this paper we consider two special cases of the ``cover-by-pairs'' optimization
problem that arise when we need to place facilities so that each customer is
served by two facilities that reach it by disjoint shortest paths.
These problems arise in a network traffic monitoring scheme proposed by
Breslau et al.\ and have potential applications to content distribution.
The ``set-disjoint'' variant
applies to networks that use the OSPF routing protocol, and the ``path-disjoint''
variant applies when MPLS routing is enabled, making better solutions possible
at the cost of greater operational expense.
Although we can prove that no polynomial-time
algorithm can guarantee good solutions for either version,
we are able to provide heuristics that do very well in practice on instances with
real-world network structure.
Fast implementations of the heuristics, made possible by exploiting
mathematical observations about the relationship between the network instances
and the corresponding instances of the cover-by-pairs problem,
allow us to perform an extensive experimental evaluation of the heuristics and what the
solutions they produce tell us about the effectiveness of the proposed monitoring scheme.
For the set-disjoint variant, we validate our claim of near-optimality
via a new lower-bounding integer programming formulation.
Although computing this lower bound requires solving the NP-hard Hitting Set problem
and can underestimate the optimal value by a linear factor in the worst case,
it can be computed quickly by CPLEX, and it equals the optimal solution value
for all the instances in our extensive testbed.

\bigskip
\bigskip

\end{abstract}


\vspace{-.25in}
\section{Introduction}\label{section:introduction}

This paper studies two new facility location problems relevant to questions of
Internet traffic monitoring and content distribution.
These problems differ from their more standard predecessors in that
each customer must be served by two facilities rather than one.
In addition, the service routes must be shortest paths and vertex-disjoint.

More specifically, suppose we are given a network modeled
as an arc-weighted, strongly connected directed graph $G=(V,A)$,
together with sets $C \subseteq F \subseteq V$,
where $C$ is the set of customer locations and $F$ is the set of
potential facility locations.
For each pair $(c,f)$, $c\in C$ and $f \in F$,
let $P(c,f)$ be the set of all shortest paths from $c$ to $f$ in $G$, which
is nonempty by our strong connectivity assumption (and would be so for
typical computer networks).

\begin{Def}
Suppose $c \in C$ is a customer location and $\{f_1,f_2\}$ is a pair of
potential facility locations in $F - \{c\}$.
Then $\{f_1,f_2\}$ {\em covers $c$ in a
pathwise-disjoint fashion} if there exist paths $p_1 \in P(c,f_1)$
and $p_2 \in P(c,f_2)$ that have no common vertex except $c$.
Such a pair {\em covers $c$ in a setwise-disjoint fashion} if {\em no} path
in $P(c,f_1)$ shares a vertex (other than $c$) with {\em any} path
in $P(c,f_2)$.
\end{Def}

\begin{Def}
A subset $F' \subseteq F$ is called a {\em pathwise-disjoint}
(respectively, {\em setwise-disjoint}) {\em cover} for $C$ if, for every $c\in C-F'$, there
is a pair $\{f_1,f_2\} \subseteq F'-\{c\}$ such that
$\{f_1,f_2\}$ covers $c$ in a pathwise-disjoint (respectively,
setwise-disjoint) fashion.
(Note that if $c \in F'$, we are assuming that $c$ covers itself, and hence the
set $C$ is a valid cover for itself.)
\end{Def}

The two problems we study are defined as follows:

\begin{Def}
In {\sc Pathwise-Disjoint Facility Location} {\sc (PDFL)}, we are given
$G$, $C$, $F$, and asked to find a pathwise-disjoint
cover of minimum size for $C$.
{\sc Setwise-Disjoint Facility Location} {\sc (SDFL)} is the same problem
except that the cover must be setwise-disjoint.
\end{Def}

The {\sc Pathwise-} and {\sc Setwise-Disjoint Facility Location}
problems arise in a variety of networking contexts.
Our primary motivation for studying them comes from a scheme proposed
in \citet{GBDS-loss08} for active monitoring of end-to-end network
performance, which we shall describe in Section \ref{section:applications}.
However, both variants have a simple alternative motivation
in terms of an idealized content distribution
problem, which we shall use to help motivate the definitions.
Suppose we wish to distribute data, such as video-on-demand,
over a network that connects our service hubs but does not provide a rapid method for
repairing link or vertex failures.
Suppose further that the service interruptions caused by such failures
would be costly to us, and that we want our distribution process to be
relatively robust against them.
A common standard of robustness is immunity to any single vertex
or link failure (as for instance might result from an accidental cable cut).
To guarantee such resilience, we would need to place multiple copies of our
data source in the network, but because of the costs of hosting such copies,
we would like to minimize the number of such hosting sites that we deploy,
rather than placing a copy of the data at each service hub.

{\sc Pathwise-Disjoint Facility Location} models this application as follows.
The network $G = (V,A)$ is the underlying fiber network linking various service hubs,
with the directions of the arcs all reversed, so that $P(c,f)$ is the set of all
shortest paths from $f$ to $c$ in the original network, rather than those from $c$ to $f$.
The set $C$ of customer locations is the set of service hubs that need access to the data.
The set $F$ of facility locations is the set of potential sites where the
data can be hosted, which we assume includes the service hubs and possibly other
network vertices.
If we assume the standard Internet default that
shortest paths should be used for routing,
the sets $P(c,f)$ now correspond to the paths in the fiber network
over which we can route content from facility $f$ to customer $c \neq f$. 
If we further assume that link capacity is not an issue, then
the pathwise-disjoint cover of minimum size for $C$ represents the
minimum-cost choice of hosting locations for our data, subject to
the constraint that no single vertex or link failure can disconnect a
(nonfailed) service hub from all the data sources.

{\sc Setwise-Disjoint Facility Location} models the variant of this application
in which we do not have control over the routing, but instead must rely on
the network to do our routing for us.
Many Internet Service Providers (ISP's)
route packets within their networks using a shortest-path
protocol such as OSPF or IS-IS.
In such protocols, packets must be routed along shortest paths, where
the weight (length) of an arc is set by the network managers so as to
balance traffic and optimize other performance metrics.
If there is more than one shortest path leaving a given router for a given
destination, then the
traffic is split evenly between the alternatives.
This can be of further help in balancing traffic, and so traffic
engineers may specifically set weights that yield multiple shortest paths
between key routers.
The actual splitting is performed based on computing hash functions
of entries in a packet's header (such as the source and destination IP addresses).
These functions are themselves randomly chosen, are subject to change
at short notice, and are typically not available to us.
Thus when there are multiple shortest paths, although contemporaneous
packets from a given router to the same destination are likely to follow
the same path, the actual route chosen may not be readily predictable.
All we know is that it must be a member of the set $P(c,f)$ of all shortest
paths from $c$ to $f$.
This means that the only way to guarantee vertex-disjoint paths
to a customer $c$ from two facility locations $f$ and $f'$
is to restrict attention to pairs $(f,f')$ such that the corresponding
shortest path sets intersect only in $c$, and consequently our problem becomes
a {\sc Setwise-Disjoint Facility Location} problem.

\medskip
In this paper, we analyze the complexity of the PDFL and SDFL problems
and propose and test algorithms for them.
A first observation is that both problems can be viewed as special
cases of {\sc Set Cover By Pairs} (SCP), first described in
\citet{HS05}.

\medskip
\noindent {\sc Set Cover By Pairs} (SCP):  
Given a ground set $U$ of elements, a set $S$ of {\em cover objects},
and a set $T$ of triples $(u,s,t)$, where $u \in U$ and $s,t \in S$,
find a minimum-cardinality covering subset $S' \subseteq S$ for $U$, where $S'$
{\em covers} $U$ if for each $u \in U$, there are $s,t \in S'$ such
that $(u,s,t) \in T$.
\medskip

PDFL and SDFL can be formulated as SCP
by taking $U = C$, $S = F$, and

$$
T  = \{(c,c,c): c \in C \} \ \cup\ \mbox{\huge \{}
\begin{array}{ll}
(c,f_1,f_2): & c \notin \{f_1,f_2\} \mbox{ and } \{f_1,f_2\} \mbox{ covers } c \mbox{ in a }\\
& \mbox{ pathwise-disjoint (setwise-disjoint) fashion}
\end{array}
\mbox{\huge \}.}
$$

We prove, subject to a complexity
assumption, that no polynomial-time algorithm can approximate
SCP to within a factor which is $2^{\log^{1-\epsilon} n}$
for any $\epsilon>0$.
The best previous hardness bound for SCP was just {\sc Set Cover}-hardness
\citep{HS05}, which implies that no $o(\log n)$ approximation algorithm can exist
unless P = NP \citep{RS97}.
We then show that SDFL is just as hard to approximate as SCP, and that PDFL is at least
as hard to approximate as {\sc Set Cover}.

These complexity results (assuming their widely-believed
hypotheses) rule out both polynomial-time heuristics that are guaranteed to find good
covers and the existence of good lower bounds on the
optimal cover size that can be computed in polynomial time.
Nevertheless, there still may exist algorithms and lower bounds that are useful
``in practice.''
In this paper we describe and experimentally evaluate our candidates for both.

We test four main heuristics.
Each uses as a subroutine a standard
randomized greedy heuristic (\greedy) that actually solves
the general SCP problem.
The first of our main algorithms, {\gr4}, is the variant of {\greedy}
that performs 400 randomized runs and returns the best solution found.
The second is a genetic
algorithm ({\genetic}) that uses {\greedy} as a subroutine.
The third and fourth, {\em Single Hitting Set} (\shs) and {\em Double
Hitting Set} (DHS), apply only in the Setwise-Disjoint case,
exploiting the graph structure in ways that are unavailable to us in 
the path-disjoint case.

The quality of the solutions that the heuristics produce can, for small
instances, be evaluated by applying a commercial optimization package
(CPLEX\texttrademark\, Version 11 in our case)
to an integer programming formulation of the derived instances of {\sc SCP}.
This is usually feasible when $|F|\le 150$, although running times grow
dramatically with graph size.
For the set-disjoint case we
introduce a new lower bound that exploits the graphical nature of our problem to
create an instance of the {\sc Hitting Set} problem (a dual to {\sc Set Cover})
whose optimal solution value can be no greater than the optimal solution value for our problem.
Although this lower bound can underestimate the optimal by a linear factor in the worst case,
and its computation requires solving an NP-hard problem, it turns out to be easily
computed using CPLEX and to give surprisingly good bounds on all our test instances.
In fact, it yielded the optimal solution value for all of them, since for each instance
at least one of our heuristics produced a solution whose value matched the bound.
Moreover, when restricted to our test instances with optimized OSPF weights, the
Hitting Set solution itself was a feasible (and hence optimal) solution to our
original problem.
For the Path-Disjoint case, our only lower bound is the comparatively weak one of
considering the variant of the problem where we drop the requirement that the paths
be shortest paths.  The optimum for the resulting problem can be computed in linear time,
but is substantially below the optimal PDFL solution, where the latter can be computed.
However, for these instances,
it at least shows us the penalty imposed by our restriction to shortest paths.
For larger instances, we can use
our {\genetic} algorithm to provide some idea of how much we might be able to 
improve on {\gr4} if time were not an issue.

A final algorithmic challenge was that of constructing the derived
{\sc SCP} instances needed by all our heuristics.
This involves exploiting shortest path graphs to determine the (often quite large)
sets of relevant triples.
Significant algorithmic ingenuity is needed to prevent this computation
from being a major bottleneck, and we will describe the mathematical observations
and algorithmic techniques that make this possible.

\subsection{Outline}
The remainder of the paper is organized as follows:
In Section \ref{section:applications} we describe the network monitoring application
that motivated our study.
In Section \ref{section:complexity} we present our complexity results.
In Section \ref{section:lbs} we present the lower bounds we use for evaluating
our cover-creating heuristics, and our algorithms for computing them.  This
includes the integer programming formulations for computing the true optima for
our problems from their Cover-by-Pairs formulations.
In Section \ref{section:heuristics} we describe the heuristics we have devised,
as well as the (nontrivial) algorithms we use to convert our problems to their
Cover-by-Pairs formulations.
Our test instances are described in Section \ref{section:instances},
and our experiments and their results are summarized in Section
\ref{section:experiments}.
We conclude in Section \ref{section:further} with a discussion of
further research directions, including preliminary
results for significant variants on our problems, such as the cases where
not all members of $C$ need be in $F$, and where facility locations have costs
and we wish to find a minimum-cost cover.

\subsection{Related Work}
This paper is, in part, the journal version of a conference paper \citep{BDD11}.
However, it contains substantial amounts of new material, including the
new SDFL lower bound and related algorithms, a faster version of the Genetic
algorithm, details of the omitted proofs, key implementation
details of our algorithms,
and more detailed experimental results and analysis, including an expansion of the
experimental testbed to contain significantly larger synthetic and real-world
instances.
The only previous work on {\sc Set Cover by Pairs}, as far as we know,
is that of \citet{HS05}, which is theoretical rather
than experimental.
That paper considers two 
applications that were significantly different from those introduced here,
and, from a worst-case point of view, much easier to approximate.
The paper also introduces a variant of the Greedy algorithm studied here
for the general {\sc SCP} problem and analyzes its (poor) worst-case behavior.

\section{Our Motivating Application: Host Placement for End-to-End Monitoring}\label{section:applications}
In this section, we describe the monitoring application of \citet{GBDS-loss08}
that motivated this paper, an application
that is more realistic than the content distribution application mentioned
in the previous section, but also more complicated.  We describe it in some
detail here to better motivate our study, and also to present a key new lemma
that is actually needed to guarantee that the proposed monitoring scheme provides
valid results.

Suppose we are an Internet Service Provider (ISP) and provide ``virtual
private network'' (VPN) service to some of our customers.
In such a service, we agree to send traffic between various locations
specified by the customer, promising to provide a certain level of
service on the connections, but not specifying the actual route the
packets will take.
(The actual routing will be done so as to optimize the utilization
of our network, subject to the promised levels of service.)
Our network is a digraph $G = (V,A)$, in which the vertices correspond to
routers and the arcs to the links between routers.
A key service quality metric is packet loss rate
(the fraction of packets on a path that fail to reach their destination).
Let $p(r_1,r_2)$ denote the probability that a packet sent from router
$r_1$ to router $r_2$ will successfully arrive.
Our goal is to obtain estimates for $p(r_i,r_j)$ for a collection of
customer paths $P_{r_i,r_j}$.
Note that, in contrast to our content distribution application, we here do
not worry about links' failing (which would cause re-routing), but merely about
their underperforming.

One way to measure the loss rate on the path in our network
from router $r_1$ to router $r_2$ is to attach extra equipment to
the routers, use the equipment at $r_1$ to send special measurement packets
to $r_2$, and use the equipment at $r_2$ to count how many of
the packets arrive.
If $N$ packets are sent and $N'$ arrive, then $N'/N$ should be
a good estimate for $p(r_1,r_2)$, assuming $N$ is sufficiently large.
Unfortunately, the process of authorizing, installing, and maintaining
the extra equipment can be time-consuming and expensive.
Thus, this scheme may not be practical in a large network with hundreds
or thousands of distinct path endpoints.
For this reason, \citet{GBDS-loss08} proposed an
alternative scheme that
may yield a substantial
reduction in the total amount of monitoring equipment needed.

\begin{figure}[h]
\vspace{-.20in}
\begin{center}
\centerline{\includegraphics[width=3.5in]{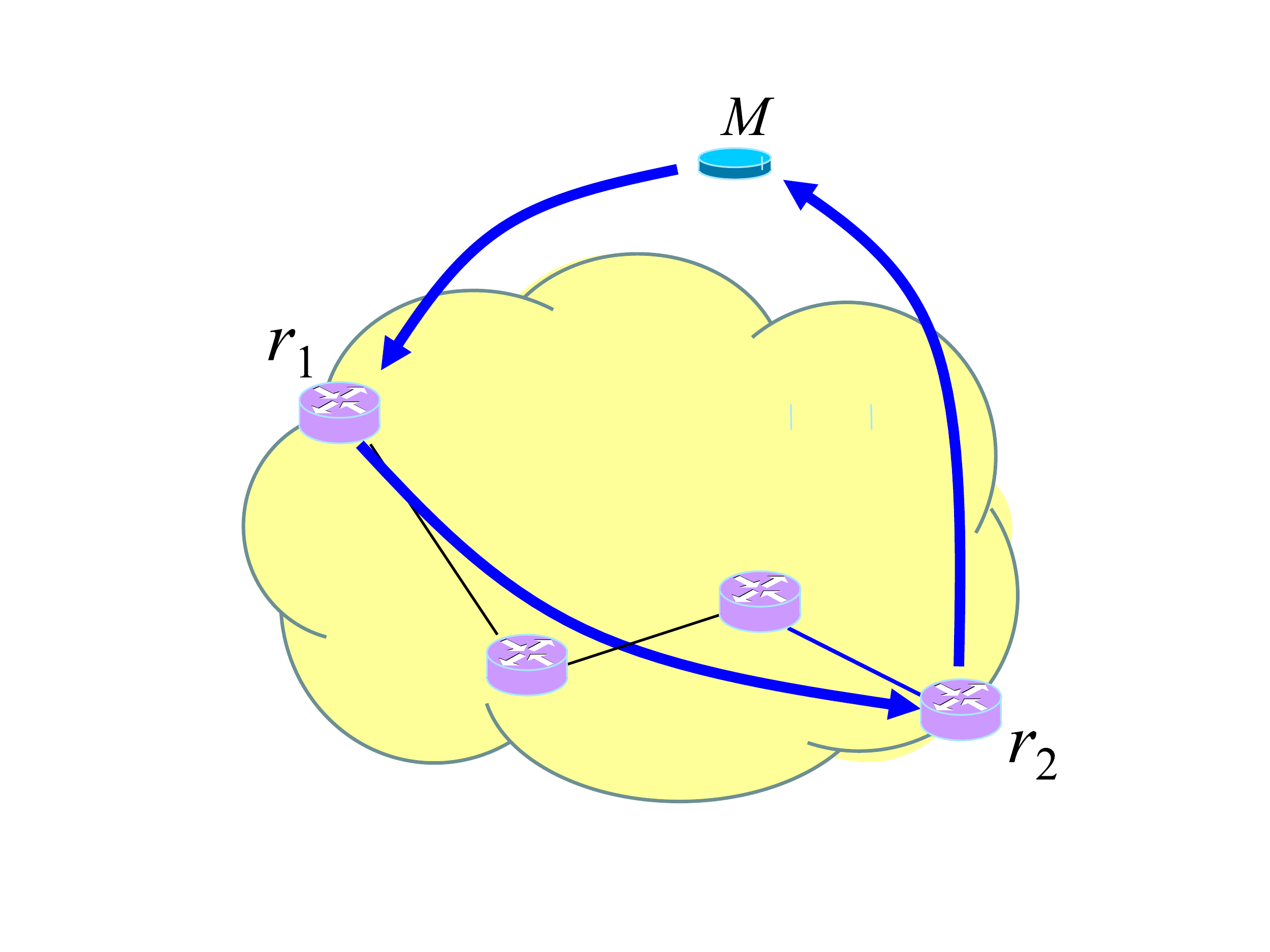}}
\vspace{-.25in}
\caption{The centralized monitoring scheme of \citet{GBDS-loss08}.\label{fig:scheme1}}
\end{center}
\vspace{-.25in}
\end{figure}

In this new scheme, all the monitoring is initiated from a single special
measurement vertex $M$, as originally proposed in 
\citet{burch:sigm05} and
\citet{breslau:inm06}.
See Figure \ref{fig:scheme1}.
To measure loss on the path from vertex $r_1$ to vertex $r_2$,
the equipment at $M$ sends a packet on a circular path that first
goes from $M$ to $r_1$ (the {\em hop-on} path),
then traverses the path from $r_1$ to $r_2$,
and finally returns from $r_2$ to $M$ (the {\em hop-off} path).
Let us make the following assumptions:
\begin{enumerate}
\item  Packets are only dropped by arcs, not by vertices.
(This is a close approximation to reality in modern-day networks,
where an arc models the system consisting of the physical wire/fiber
connecting its endpoints, together with the line card at each of its ends.)
\item The three paths $P_{M,r_1}$, $P_{r_1,r_2}$,
and $P_{r_2,M}$ are pairwise arc-disjoint.
(As we shall show below, this will typically be true under
shortest-path routing.)
\item Loss rates on different arcs are independent of each other.
(This is somewhat less realistic, but is approximately true except
in heavily-loaded networks.)
\end{enumerate}
Then if $N$ packets are sent on the circular path $P_{M,r_1,r_2,M}$,
the expected number $N'$ of packets successfully
making the roundtrip will be $N' = Np(M,r_1)p(r_1,r_2)p(r_2,M)$.
Thus if we measure $N'$ and have good estimates for $p(M,r_1)$ and $p(r_2,M)$,
we will have the estimate
$$
p(r_1,r_2) = \frac{N'/N}{p(M,r_1)p(r_2,M)}\:.
$$
Thus we have reduced the problem of measuring the loss rates for a
collection of paths between arbitrary vertices to that of measuring the loss rates
on a collection of round-trip paths and estimating the
loss rates for a collection of hop-on and hop-off paths,
all of which either begin or end at $M$.

\citet{breslau:inm06} proposed that these loss rates for a given
path endpoint $r$ be estimated
by sending packets along an $(M,r,M)$ circuit and, if, here, $N$ packets were sent
and $N'$ received, concluding that $p(M,r) = p(r,M) = \sqrt{N'/N}$.
Unfortunately, this assumes that Internet performance is symmetric, which
it definitely is not.
A quite accurate way to measure the loss rates would of course be
to put equipment at both ends of each of the hop-on and hop-off paths,
but this method would
require installing equipment at just as many routers as in the original
scheme for measuring the $P_{r_1,r_2}$ paths directly -- indeed at
one more vertex, since now we need equipment at $M$.
{\sc Setwise-} and {\sc Pathwise-Disjoint Facility Location}
arise in the context of a ``tomographic'' method proposed by
\citet{GBDS-loss08} for estimating loss rates on hop-on and hop-off paths in a
potentially much more efficient fashion.

\begin{figure}
\begin{center}
\vspace{-.35in}
\centerline{\includegraphics[width=3.5in]{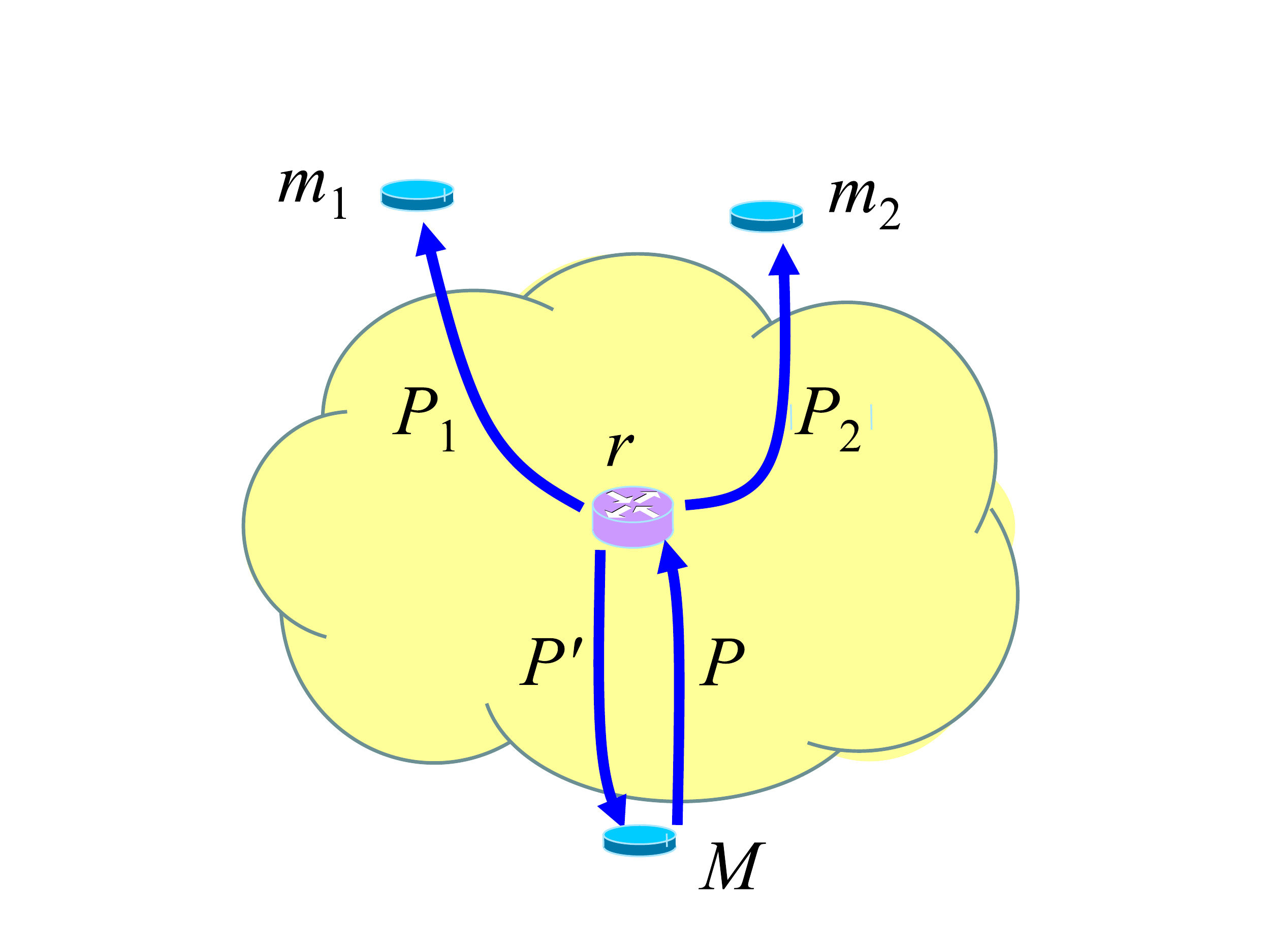}}
\vspace{-.1in}
\caption{Scheme of \citet{GBDS-loss08}
for measuring loss rate of hop-on and hop-off paths.\label{fig:scheme1a}}
\end{center}
\vspace{-.25in}
\end{figure}

In terms of the facility location problems, the set $C$ of ``customer'' vertices will
consist of the endpoints of the original paths whose loss rates we wish
to estimate.
The set $F$ of ``facility locations'' will be these plus those additional vertices
that are capable of hosting monitoring equipment.
In this context, we will call $F$ the set of
(potential) monitoring vertices.
Assuming as before that we are in a shortest-path routing regime
such as OSPF or IS-IS, the set $P(r,m)$
is the set of all legal routes from $r$ to $m$.

We have already observed that, if we install equipment at $r$ itself,
it is straightforward to estimate the loss rates $p(M,r)$
(for the hop-on path to $r$)
and $p(r,M)$ (for the hop-off path from $r$) --
simply send packets between $M$ and $r$ and counting the number
of them that successfully arrive at their destinations.
Suppose $r$ is a path endpoint without equipment and $(m_1,m_2)$ is a pair of monitoring vertices
that cover $r$ in a pathwise-disjoint fashion.
We will now explain how, by installing monitoring equipment
at $m_1$ and $m_2$, we can estimate both loss rates.
See Figure \ref{fig:scheme1a}.
Assuming we are allowed to specify the routing paths from $r$ to
$m_1$ and $m_2$, the fact that $m_1$ and $m_2$ cover $r$ in a
pathwise-disjoint
fashion means that we can pick legal routing paths $P_1$ and $P_2$ from $r$ to
$m_1$ and $m_2$, respectively, that are vertex-disjoint except
for $r$ (and hence arc-disjoint).
Moreover, as we shall see, we also have that the two paths $P_1$ and $P_2$
are arc-disjoint from the path $P$ from $M$ to $r$, which itself is arc-disjoint from
the path $P'$ from $r$ to $M$ (under reasonable assumptions).

This is a consequence of the following lemma, whose additional
assumptions have to do with the arc weights used by OSPF and IS-IS 
in their shortest path computations.
These weights are set by traffic engineers to help balance
traffic loads and normally obey certain restrictions.
First, they are positive integers.
Second, in practice networks are typically {\em symmetric} directed graphs, in
that the digraph contains an arc $(a,b)$, then it must also contain
arc $(b,a)$, and we assume that the weights $w$ for our digraph are
themselves symmetric, in that for every arc $(a,b)$,
we have $w(a,b) = w(b,a)$.
In real world networks, the weights typically are symmetric, and even networks
with asymmetric weights have very few arcs where $w(a,b) \neq w(b,a)$.

\begin{lemma}\label{sptheo}
Suppose we are given a symmetric directed graph $G=(V,A)$, a weight
function $w$ on the arcs that is symmetric and positive, and three
vertices $a,b,c$.  If $P_{a,b}$ and $P_{b,c}$ are shortest-weight paths
in this digraph from $a$ to $b$ and $b$ to $c$, respectively, then
they are arc-disjoint.
\end{lemma}

\smallskip
\noindent
\proof
Suppose that, contrary to the Lemma, $P_{a,b}$ and $P_{b,c}$ share
a common arc $(x,y)$.
Then the path $P_{a,b}$ can be broken up into a path
$P_{a,x}$ followed by arc $(x,y)$ followed by a path $P_{y,b}$,
where $a$ and $x$ may possibly be the same node, as may $y$ and $b$.
Similarly the path $P_{b,c}$ can be broken up into a path $P_{b,x}$
followed by arc $(x,y)$ followed by a path $P_{y,c}$.
See Figure \ref{fig:theorem1}.

\begin{figure}
\begin{center}
\includegraphics[width=4in]{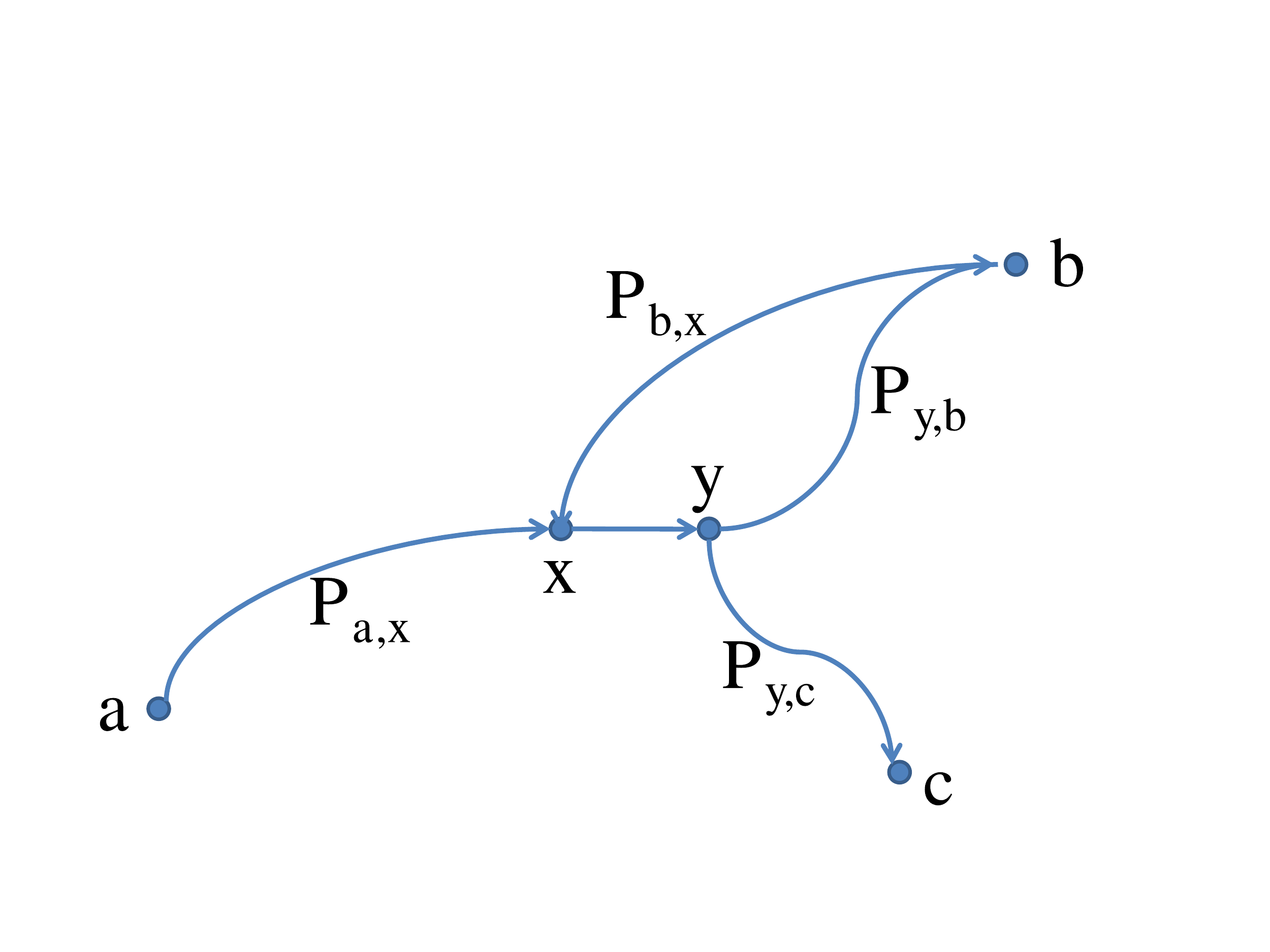}
\caption{Illustration for proof of Lemma \ref{sptheo}.\label{fig:theorem1}}
\end{center}
\end{figure}

Note that each of the subpaths of these paths must themselves be shortest
paths between their endpoints since the overall path is itself a
shortest path.
For instance the path $P_{x,b}$ that is the concatenation of $(x,y)$
with $P_{y,b}$ is a shortest path from $x$ to $b$.
Let us extend the definition of $w$ to paths $P$ by letting $w(P)$ be the
sum of the weights of the arcs in $P$.
Then, for instance, $w(P_{x,b}) = w(x,y) + w(P_{y,b})$.
If we let $P_{b,y}$ be the path from $b$ to $y$ obtained by reversing
the path $P_{y,b}$, then, since weights are symmetric, we must have
$w(P_{b,y}) = w(P_{y,b}) = w(P_{x,b}) -  w(x,y) =  w(P_{b,x}) -  w(x,y)$.

But now consider the path from $b$ to $c$
obtained by concatenating the two paths $P_{b,y}$ and $P_{y,c}$.
It has length $w(P_{b,y}) + w(P_{y,c}) = w(P_{b,x}) -  w(x,y) + w(P_{y,c})$.
On the other hand, by hypothesis, the shortest path from $b$ to $c$ has length
$ w(P_{b,x}) +  w(x,y) + w(P_{y,c})$, and is hence $2w(x,y)$ longer than
the length of this $P_{b,y}$ plus $P_{y,c}$ path.
Since by hypothesis $w(x,y) > 0$, this is a contradiction.
\qed

The basic idea of the technique of
\citet{GBDS-loss08} for estimating the loss rate $p(M,r)$ using these paths
is to send multicast packets
from $M$ to $r$ along path $P$, replicate them at $r$, and 
then send the copies along paths $P_1$ and $P_2$ to $m_1$ and $m_2$, respectively.
After this, $m_1$ and $m_2$ report back to $M$ (using a
guaranteed-delivery service such as TCP) as to which packets arrived.
Based on this information, $M$ estimates $p(M,r)$.
The loss rate $p(r,M)$ can
be estimated by sending packets along
the $(M,r,M)$ loop and counting the number that arrive back at $M$,
using the fact that the loss rate for the loop should be $p(M,r)p(r,M)$.
(We note that a result like Lemma \ref{sptheo} is needed if this method
is to provide reliable estimates, a fact not observed in \citet{GBDS-loss08},
which contained no such result.)

This scheme may require two monitoring hosts to measure the
hop-on and hop-off rates for a path endpoint $r$, rather than the single one that
would be required if we placed the monitoring equipment at vertex $r$ itself.
However, the scheme has the potential advantage that a given monitoring vertex
can be re-used to handle many different path endpoints.
Thus there could be a substantial
net overall savings in the total number of monitoring
vertices used, and hence in equipment and operational cost.

As stated,
the problem of finding a minimum-sized set of monitoring vertices at which
to place equipment so that we can estimate loss rates for all hop-on and
hop-off paths is simply our original {\sc Pathwise-Disjoint Facility
Location} problem.
In practice, however, we will most-likely
have to rely on the ISP's routing protocol
(OSPF or IS-IS) to deliver our packets, and so, as with our the first application,
will face the {\sc Setwise-Disjoint Facility Location} problem.

It should be noted that, in contrast to that first application, the necessity
for {\em vertex}-disjoint paths from $r$ to $m_1$ and $m_2$,
rather than simply arc-disjoint paths, is less clear, since
by the previous lemma we can only guarantee that these paths are
arc-disjoint from the path from $M$ to $r$.
This is a meaningless distinction in the Setwise-Disjoint case, however,
in light of the following lemma.

\begin{lemma}\label{avtheo}
Suppose $P(c,f)$ and $P(c,f')$ are the sets of all shortest paths from
vertex $c$ to vertices $f$ and $f'$, respectively, in a given digraph $G$.
Then no path in $P(c,f)$ shares an arc with any path in $P(c,f')$ if
and only if no path in $P(c,f)$ shares a vertex
other than $c$ with any path in $P(c,f')$.
\end{lemma}

\noindent
\proof
The no-shared-vertex case trivially implies the no-shared-arc case, so
let us assume that there are no shared arcs and argue that there are
also no shared vertices other than $c$.
Suppose there were a shared vertex $x \neq c$.
Then there is a shortest path $P$ from $c$ to $f$ that is the concatenation
of a path $P_{c,x}$ from $c$ to $x$ with a path $P_{x,f}$
from $x$ to $f$, both themselves shortest paths,
and a shortest path $P'$ from $c$ to $f'$ that is the concatenation
of a path $P_{c,x}'$ from $c$ to $x$ with a path $P_{x,f'}'$ from $x$
to $f'$, both themselves shortest paths.
But then the paths $P_{c,x}$ and $P_{c,x}'$ must have the same length.
Thus the path $P''$ that is the concatenation of $P_{c,x}$ with $P_{x,f'}'$
is a shortest path from $c$ to $f'$, and must be contained in $P(c,f')$.
Moreover, since $x \neq c$, the path $P_{c,x}$ contains at least one
arc $a$.
Thus the path $P'' \in P(c,f')$ shares the arc $a$ with the path
$P \in P(c,f)$, a contradiction.
\qed

A detailed description of the implementation of this scheme
and the formulas used for estimating $p(M,r)$ and $p(r,M)$ is
presented in \citet{GBDS-loss08}.

\section{Complexity}\label{section:complexity}

In this section we investigate the
computational complexity of {\sc Pathwise-} and
{\sc Setwise-Disjoint Facility Location}, covering both general hardness
results and a polynomial-time solvable special case.

\subsection{Hardness of Approximation}\label{section:hardness}

We first observe that the more general {\sc Set Cover by Pairs} problem is
not only NP-hard, but also strongly inapproximable in the worst-case.
Let $n = |U|$.
\citet{HS05} observed that SCP is at least as hard to approximate
as {\sc Set Cover}, which cannot be approximated to within a $o(\log n)$ factor
unless P = NP \citep{RS97}.
We can prove much a stronger inapproximability result
(albeit with a slightly stronger complexity assumption).

\begin{theorem} \label{scp-hard}
If ${\rm NP} \not \subseteq {\rm DTIME} (n^{O({\rm polylog}(n))})$,
no polynomial-time algorithm can be guaranteed to find a solution to
SCP that is within a factor of $2^{\log^{1-\epsilon} n}$ of optimal
for any $\epsilon > 0$.
\end{theorem}

\medskip
\noindent
\begin{proof}
The theorem follows via an approximation-preserving transformation
from the {\sc MinRep} problem of Kortsarz, who showed the above
inapproximability bound to hold for {\sc MinRep} \citep{Kortsarz01}.
The reader may readily confirm that, for the type of approximation
bound being proved here, it suffices for the transformation to
provide a one-to-one correspondence between solutions
for the source and target instances (and their values) and take
polynomial time (which also causes the blow-up in the size of
the constructed instance to be bounded by a polynomial).

In {\sc MinRep}, we are given a bipartite graph $G$ with vertex set $V=A \cup B$
and edge set $E$, where $|A|=|B|=kq$ for positive integers $k,q$,
together with partitions of the vertex sets on each side of the
bipartite graph into $k$ groups of $q$ vertices, $A_1$ through $A_k$
on the left and $B_1$ through $B_k$ on the right.
We are also given an integer $K$.
We ask whether there is a subset $V' \subseteq V$ with $|V'| \leq K$
such that for every pair $(A_i,B_j)$, $1 \leq i,j \leq k$, if
$G$ contains an edge between a vertex in $A_i$ and one in $B_j$,
then so does the subgraph induced by $V'$.

We transform {\sc MinRep} to SCP by letting the items to be covered
be the pairs $(A_i,B_j)$ where $G$ contains an edge between a member
of $A_i$ and a member of $B_j$.
The set of covering objects is $V$, with
the item $(A_i,B_j)$ covered by all pairs $\{u,v\} \subseteq V$
in which $u \in A_i$,
$v \in B_j$, and $(u,v) \in E$.
There is then a one-to-one correspondence between SCP cover sets
and subsets $V'$ meeting the {\sc MinRep} requirements,
with the two sets having the same size.
Hence any approximation guarantee for SCP implies an equally good
guarantee for {\sc MinRep}, and so SCP must be at least as hard
to approximate.
\end{proof}

\begin{theorem}  \label{equivalence}
The special cases PDFL and SDFL retain much of the complexity of SCP:
\begin{enumerate}
\item SDFL is at least as hard to approximate as SCP.
\item PDFL is at least as hard to approximate as {\sc Set Cover}.
\end{enumerate}
\end{theorem}

\begin{proof}
We shall initially prove the results without the restriction
that $C \subseteq F$, and then
sketch how the proofs can be modified (by appropriate replication) to
enforce the restriction.

For Claim 1, we will show how to transform an arbitrary instance of SCP into
an equivalent instance of SDFL.
Suppose the given instance of SCP consists of sets $U = \{u_1,\ldots ,u_p\}$
and $S = \{s_1,\ldots ,s_q\}$
and relation $T \subseteq U \times S \times S$.
By the definition of SCP,
we may assume without loss of generality that if $(u,s_i,s_j) \in T$,
then $i \leq j$.
In the corresponding instance of SDFL, there are four types of vertices.
\begin{enumerate}
\item For each $u \in U$, a customer vertex $c_u \in C$.
\item For each $u \in U$, a {\em connector} vertex $x_u$.
\item For each $s \in S$, a facility location vertex $f_s \in F$.
\item For each pair $(s_i,s_j)$, $1\leq i < j \leq q$,
a {\em pair vertex} $v_{s_i,s_j}$.
\end{enumerate}
There are four types of edges, all having length 1.
(The constructed digraph is symmetric, so instead of constructing
both $(a,b)$ and $(b,a)$, we just construct one undirected edge $\{a,b\}$.)
\begin{enumerate}
\item For each customer vertex $c_u \in C$, an edge $\{c_u,x_u\}$.
\item For each connector vertex $x_u$ and each facility location
vertex $f_s$, an edge $\{x_u,f_s\}$.
\item For each pair $(s_i,s_j)$, $i < j$, an edge between
$c_u$ and $v_{s_i,s_j}$ {\em unless} $(u,s_i,s_j) \in T$.
\item For each pair $(s_i,s_j)$, $1\leq i < j \leq q$, the two edges
$\{v_{s_i,s_j},f_{s_i}\}$ and $\{v_{s_i,s_j},f_{s_j}\}$.
\end{enumerate}

We claim that $F' \subseteq F$ is a feasible cover for the constructed
SDFL instance if and only if $S' = \{ s: f_{s} \in F'\}$ is a feasible cover
for the original SCP instance.
Note that in the constructed digraph, there is a path of length two between
each customer vertex $c_u$ and each facility location vertex $f_{s_i}$, i.e.,
$\langle c_u, x_u, f_{s_i}\rangle$, and this is the shortest possible path.
Thus $P(c_u,f_{s_i})$ consists precisely of this path, together with all
paths $\langle c_u,v_{s_i,s_j},f_{s_i}\rangle$
where $i \leq j$ and $(u,s_i,s_j) \notin T$
and all paths $\langle c_u,v_{s_j,s_i},f_{s_i}\rangle$
where $i > j$ and $(u,s_j,s_i) \notin T$.
Thus the only vertex that a path in $P(c_u,f_{s_i})$
can possibly share with a path in $P(c_u,f_{s_j})$, $i \leq j$,
other than $c_u$ itself, is
$v_{s_i,s_j}$, which will be in shortest paths from $c_u$ to
$f_{s_i}$ and  to $f_{s_j}$ if and only if $(u,s_i,s_j) \notin T$.
Hence a pair $\{s_i,s_j\}$ will cover an element $u$ in the SCP instance
if and only if the pair $\{f_{s_i},f_{s_j}\}$
jointly covers the customer vertex $c_u$ in the constructed SDFL instance,
so the optimal solution values for the two instances coincide.

To get an equivalent instance in which $C \subseteq F$,
we replicate the above construction $|S|+1$ times, while, for each $s \in S$, identifying
the $|S|+1$ copies of vertex $f_s$.
We then add the $|S|+1$ copies of each customer vertex $c_u$ to $F$ to obtain
a new set $F_+$ of facilities, one that now contains all the customer vertices
in our expanded instance.
Let $F'$ be an optimal solution to our original instance.
Note that it remains an optimal solution to the expanded instance.
Thus if $F_+'$ is an optimal solution to this new instance, we must have
$|F_+'| \leq |F'| \leq |F| = |S| < |S|+1$.
Thus at least one of the $|S|+1$ copies of our original instance fails
to have any of its customer vertices in $F_+'$.
This implies that $F \cap F_+'$
must be a feasible solution for that copy, and so $|F_+'| \geq |F'|$.
But this implies that $|F_+'| = |F'|$.
So the optimal solution value for our expanded instance still equals that
for the original instance of SCP.
Moreover, the expanded instance can be constructed in polynomial time and has size within
a polynomial factor of the size of the original SCP instance.
Hence the inapproximability result for SCP carries over to SDFL.

\smallskip
For Claim 2, we again prove the result without the restriction that $C \subseteq F$.
The extension to the case where the restriction holds follows by
a replication construction similar to that given in the proof of Claim 1.
Our transformation is from {\sc Set Cover}, which we already
know cannot be approximated to within a $o(\log n)$ factor unless P = NP
\citep{RS97}.
Given a {\sc Set Cover} instance, we construct an instance
of PDFL whose optimal solution differs from the optimal solution to
the {\sc Set Cover} instance by at most 1.
Thus a polynomial-time $o(\log n)$ approximation for our problem would imply one
for {\sc Set Cover}.

An instance of {\sc Set Cover} consists of a
ground set $U = \{u_1,\ldots,u_n\}$ and a collection
${\cal C} = \{C_1,\ldots,C_m\}$ of subsets of $U$.
Assume without loss of generality that $\cup_{C \in \cal{C}}C = U$.
Our constructed graph has four types of vertices:
\begin{enumerate}
\item For each $i$, $1 \leq i \leq n$, a customer vertex $c_i \in C$.
\item For each $i$, $1 \leq i \leq n$, a {\em connector} vertex $x_i$.
\item For each $j$, $0 \leq j \leq m$, a potential facility vertex $f_j \in F$.
Note that $f_0$ is an added facility, not corresponding to any member
of ${\cal C}$.
\item A {\em universal connector} vertex $x_0$.
\end{enumerate}
There are four types of (undirected) edges, all having length 1:
\begin{enumerate}
\item For $1 \leq i \leq n$, an edge $\{c_i,x_i\}$.
\item For each $i,j$, $1 \leq i \leq n$ and $1 \leq j \leq m$,
such that $u_i \in C_j$, an edge $\{x_i,f_j\}$.
\item For each $i$, $1 \leq i \leq n$, an edge $\{c_i,x_0\}$.
\item For each $j$, $0 \leq j \leq m$, an edge $\{x_0,f_j\}$.
\end{enumerate}

Note that for each pair $i,j$, $1 \leq i \leq n$ and $1 \leq j \leq m$,
there is a path of length two from $c_i$ to $f_j$ through $x_0$, and
this is the shortest possible path.
Thus for each such $i,j$, $P(c_i,f_j)$ contains only paths of length two.
More precisely, by construction, if $u_i \notin C_j$ it contains just the path
$\langle c_i,x_0,f_j\rangle$, and otherwise it contains
that path plus the path $\langle c_i,x_i,C_j\rangle$.

Suppose first that $\cal{C}'$ is a feasible solution to our {\sc Set Cover}
instance.  Then it is easy to see that $\{f_0\}\cup\{f_j: C_j \in \cal{C}'\}$
is a feasible solution for our problem.
Each $c_i$ is covered by $f_0$ and an $f_j$ such that $u_i \in F_j$, which
must exist because $\cal{C}'$ is a cover.
Thus the optimum for our constructed instance is at most one more than
the {\sc Set Cover} optimum.

Conversely, suppose that $F' \subseteq F$ is a feasible solution for
our constructed instance.
For each customer vertex $c_i$, there must be at least one vertex in $F'$
that does not equal $f_0$, and so $F' -\{f_0\}$ must be a set cover for
the original instance.
Thus the optimum for the original instance is at most the optimum for
our constructed instance.
As in the proof of Claim 1, we can now, by replication, construct an instance of PDFL where
all customers are facilities, and the optimum value remains unchanged.
Consequently, the two optima are within one of each other, which implies that
a polynomial-time algorithm with $o(\log n)$ approximation guarantee for PDFL
would yield one for {\sc Set Cover} as well.
\end{proof}

\subsection{The Special Case of Trees}\label{sec:trees}
In this section, we consider networks derived from trees,
where for a given tree $T = (V,E)$ the corresponding network is the symmetric
directed graph
$G = (V,A)$, with both $(u,v)$ and $(v,u)$ in $A$ whenever $\{u,v\} \in E$.
We first note that, for such instances, the two problems PDFL and SDFL
coincide.
This follows from the fact that
for every ordered pair $v,w$ of vertices, there is exactly
one simple path from $u$ to $v$, which is of course the shortest path,
and so the shortest path is unique.
More interestingly, the problems can be solved in polynomial (in fact, linear)
time.
Indeed, the optimal solution has a simple combinatorial characterization,
and can be constructed by a simple tree traversal.

\begin{theorem} \label{treetheo}
The following procedure constructs an optimum cover under PDFL.
\begin{enumerate}
\item While $T$ contains a leaf vertex $v$ that is not a customer, delete $v$.
\item Return the set $F^*$ of leaves of the residual tree as our cover.
\end{enumerate}
\end{theorem}
\begin{proof}
Note that in the residual tree, all leaves are customers and hence facility
locations, so $F^* \subseteq F$, as required.
Every customer in $F^*$ is already covered by itself.
If there is any customer $v \notin F^*$, it must be an internal vertex of
the residual tree, and hence there must be at least two edges incident on $v$
in the residual tree.
Each such edge must lie on a path to a leaf of the residual tree, and those
paths can contain no common vertex other than $v$ because we are dealing
with a tree.
But all leaves of the residual tree are in $F'$, so this means that $v$
has vertex-disjoint shortest paths to at least two members of $F'$ and hence
is covered by $F'$.
\end{proof}

\section{Algorithms}\label{section:lbs}
The complexity results of the previous section rule out (assuming widely-believed
hypotheses) both polynomial-time heuristics that are guaranteed to find good
covers and the existence of good lower bounds on the
optimal solution value that can be computed in polynomial time.
Nevertheless, there still may exist algorithms and lower bounds that are useful
``in practice.''
In this section and the next, we describe our candidates for both, beginning with lower bounds,
which provide us with a standard of comparison for evaluating the quality of the
solutions our heuristics provide.

\subsection{Optimal Solutions Using Mixed Integer Programming.}\label{section:MIP}
The best standard of comparison is, of course, the optimal solution value.
It is feasible to compute this for small instances using mixed integer
linear programming, by essentially viewing our instances of PDFL and SDFL
as {\sc Set Cover by Pairs} instances, modified to account for the fact
that a customer, in its role as a facility, can cover itself.
This allows us to leverage general purpose commercial software for solving
our problems.

Our MIP formulation is simple and straightforward.
It has a zero-one variable $x_f$ for each facility $f \in F$,
where $x_f=1$ if $f$ is chosen for our cover.
In addition, we have a real nonnegative variable $y_{f,f'}$ for each
pair $\{f,f'\}$ of facilities, subject
to the constraints that $y_{f,f'}\le x_f$
and $y_{f,f'}\le  x_{f'}$.
Note that these together imply that $y_{f,f'}$ can be positive only
if both $x_f$ and $x_{f'}$ equal 1 (i.e., are in the chosen cover).
To guarantee that the chosen set is a cover,
we have the following constraints, one for each facility $c \in C$.
\begin{equation}\label{ILP1}
x_c + \sum_{f,f':\{f,f'\} \mbox{ covers } c} y_{f,f'}\ge 1
\end{equation}

\noindent
where ``covers'' is interpreted as ``covers in a path-disjoint fashion'' for PDFL
and as ``covers in a set-disjoint fashion'' for SDFL.
The goal of our MIP is to minimize $\sum_{f \in F} x_f$.

We derive the MIP instances from the corresponding PDFL/SDFL instances using the
``triple generation'' algorithms described at the end of this section, and then
attempt to solve them using the version 11.0 CPLEX\texttrademark\ MIP solver.
This MIP approach proved
practical for a wide range of ``small'' instances, enabling us to find optimal
solutions to all but two of the instances in our test set with $|F| \leq 150$.
Our detailed experimental results will be summarized in
Section \ref{section:experiments}.

\subsection{A Hitting Set Lower Bound}\label{section:HSLB}
For larger instances of the set-disjoint variant of our problem
(SDFL), our standard of comparison is derived from a new lower
bound, based on the construction of a special instance of the NP-hard {\sc Hitting Set}
problem that CPLEX finds it particularly easy to solve.
The bound depends on the following simple lemma.

\begin{lemma}\label{neighborlemm}
Suppose $G = (V,A)$ is an arc-weighted directed graph, and $u,v,w$ are distinct
vertices in $V$.  Let $P_{v}$ be a shortest path from $u$ to $v$, and let $x$ be the
first vertex encountered on this path.
If there is a shortest path $P_{w}$ from $u$ to $w$ that shares a common vertex other than $u$
with $P_{v}$, then there is a shortest path $P_{w}'$ from $u$ to $w$ that contains $x$.
\end{lemma}
\begin{proof}
Let $y$ be the first vertex in $P_{w}$, other than $u$, that is also in $P_{v}$.
If $y=x$, we are done, so suppose not.
Divide $P_v$ into subpaths $P_{v,1}$ from $u$ to $y$ and $P_{v,2}$ from
$y$ to $v$.
Note that $x$ remains the first vertex encountered in $P_{v,1}$.
Similarly divide $P_{w}$ into subpaths
$P_{w,1}$ from $u$ to $y$ and $P_{w,2}$ from $y$ to $w$.
Note that all these paths must be shortest paths from their sources to their
destinations.
Thus, in particular, $P_{v,1}$ and $P_{w,1}$ must have the same length since
they are both paths from $u$ to $y$.
Hence, the concatenated path $P_w' = P_{v,1} P_{w,2}$ must be a shortest path
from $u$ to $w$, and it contains $x$.
\end{proof}

For a given customer $c$, let $N(c)$ denote the set of all vertices $x$ such that
$(c,x) \in A$ -- the {\em neighbors} of $c$.
For each facility $f \neq c$,
let $N(c,f)$ be that subset of nodes of $N(c)$ that are included on shortest paths from
$c$ to $f$.
As a consequence of Lemma \ref{neighborlemm}, a pair $\{f_1,f_2\}$ will
cover $c$ in a setwise-disjoint fashion if and only if $N(c,f_1)\cap N(c,f_2)$ is
empty.
Thus a necessary (although not sufficient)
condition for $F' \subseteq F$ to contain a pair $\{f_1,f_2\}$
that covers $c$ is that $\bigcap_{f \in F'}N(c,f) = \emptyset$.

We can use this observation to construct an integer program
whose feasible set contains all legal
setwise-disjoint covers, and hence whose optimal solution
provides a lower bound on the size of
the optimal set-disjoint cover.
As in the MIP for an optimal solution, we have zero-one variables $x_f$ for
all $f \in F$, with $x_f = 1$ meaning that facility $f$ is in our chosen cover,
and we wish to minimize $\sum_{f \in F} x_f$.
Now, however, these are the only variables, and we have a different type of
constraint, one for each pair of a customer $c$ and a neighbor $x$ of $c$:
\begin{equation}\label{LBIP}
x_c + \sum_{f\neq c \mbox{ and }x \notin N(c,f)}x_f\ge 1.
\end{equation}

\begin{figure}[t]
\vspace{-.1in}
\begin{center}
\includegraphics[width=3.5in]{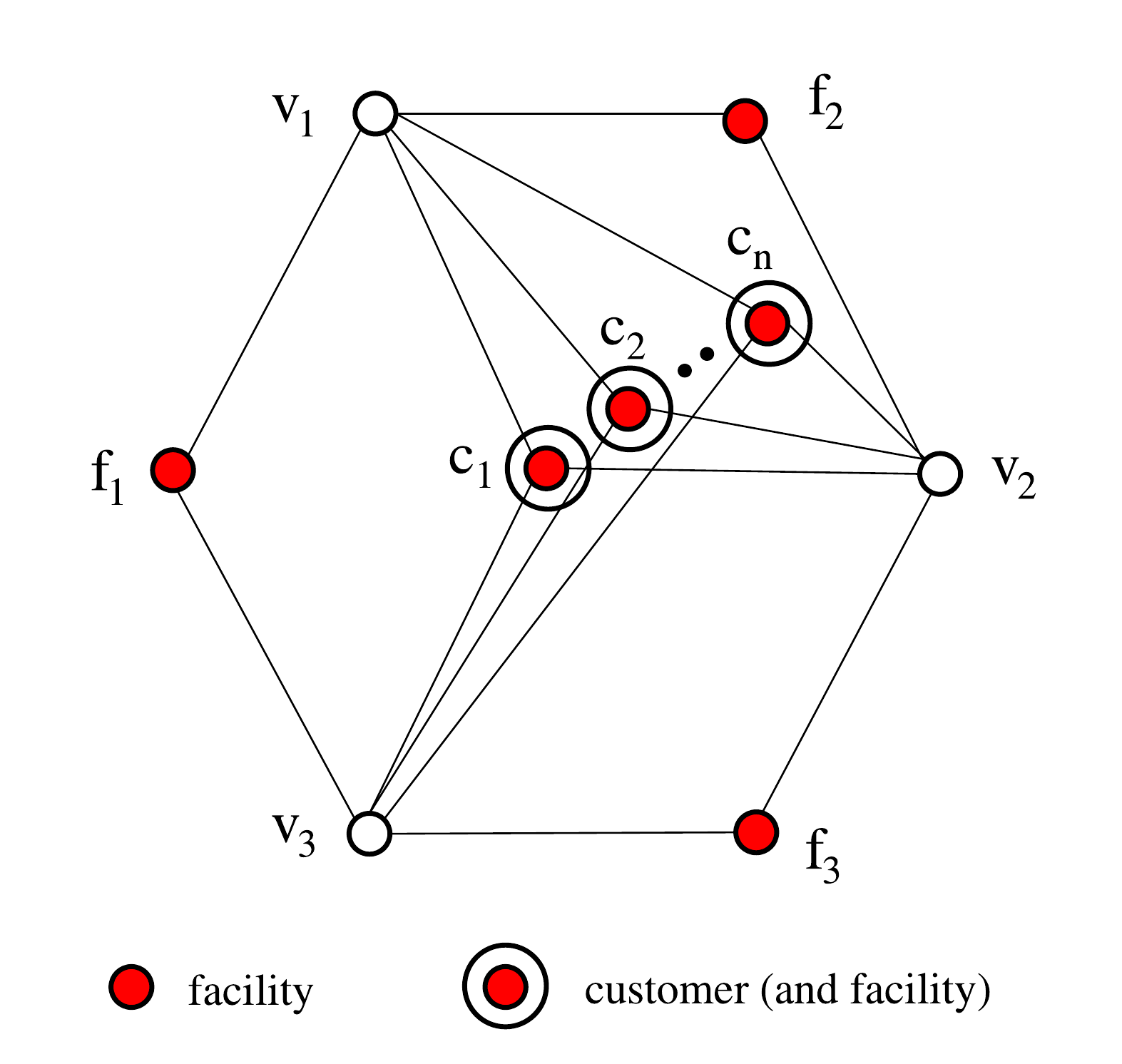}
\vspace{-.1in}
\caption{Instances where the HSLB underestimates the optimal cover size by a factor of $n/3$.}\label{fig:HSLB}
\end{center}
\vspace{-.3in}
\end{figure}

This IP can be viewed as a ``Hitting Set'' problem, where a facility
``hits'' a pair $(c,x)$ if $x \notin N(c,f)$ or if $f=c$.
For this reason, we call the optimal solution to the IP the {\em hitting set lower bound}
(HSLB).
It is easy to see that computing this bound is NP-hard.
Thus its computation
is likely to require exponential time in the worst case, and
our SDFL hardness of approximation result, which only applies to
polynomial-time algorithms, is not applicable.
Nevertheless, the Hitting Set Lower Bound can be quite bad:  In the worst
case it can underestimate the optimal solution value by a linear factor.
See Figure \ref{fig:HSLB}, which depicts a scheme for instances where the HSLB
is too low by a factor of $n/3$.
In this figure there are $n+3$ facility vertices, $f_1$, $f_2$, $f_3$, plus the
customer vertices $c_1$ through $c_n$, together with three additional vertices
($v_1$, $v_2$, $v_3$) that are the neighbors of the customer vertices.
Each (undirected) edge in the figure represents a pair of oppositely directed arcs
of length 1.
Note that for each neighbor vertex $v$, there is one of the $f_i$ whose shortest
paths from the customer vertices $c_i$ do not contain $v$.
Thus, for $n>3$, the solution to the HSLB integer program is the set $\{f_1,f_2,f_3\}$, for a
lower bound of 3.
On the other hand, for each $c_i$ and each pair of facilities
$\{f_j,f_k\} \subseteq \{f_1,f_2,f_3\}$, we have $N(c_i,f_j) \cap N(c_i,f_k) \neq \emptyset$
since it contains one of the neighbor vertices.
Thus no pair of these vertices covers any customer vertex.
In addition, for any customer vertex $c_h$ other than $c_i$, {\em all} three neighbor
vertices are on shortest paths from $c_i$ to $c_h$, so no other customer vertex can
help cover $c_i$ either.
Thus all the customer vertices must cover themselves, and hence the optimal cover
has size $n$. 

Fortunately, the kind of structure occurring in Figure \ref{fig:HSLB}, where
$|N(c,f)|/|N(c)| = 2/3$ for all customer/facility pairs, does not tend to occur
in practice.
Because of the way networks are designed and the way OSPF weights are typically set,
the number of shortest path ties is typically limited, and the ratios $|N(c,f)|/|N(c)|$
are much smaller.
In part because of this, the hitting set lower bound is quite good in practice.
Indeed, as we shall see in Section \ref{section:experiments},
the HSLB equaled the optimal solution value for all our SDFL
test instances.
The solution to the IP was also often a feasible (and hence optimal) solution
to the original SDFL instance.
Moreover, the hitting set instances themselves turn out to be relatively
easy for CPLEX to solve, so we were able
to obtain HSLBs for all our SDFL test instances in reasonable time,
and usually get optimal solutions as a welcome side-effect..
In light of these observations, we will also include our HSLB code in
our comparisons of solution-finding heuristics,
albeit one that either produces an optimal solution or no solution at all.

\subsection{An Unconstrained Path-Disjoint Lower Bound}\label{sec:pdlb}

The previous lower bound applied only to the set-disjoint variant of our problem,
and required a linear programming package to compute.
For the path-disjoint variant (PDFL), the best polynomial-time computable lower
bound we currently have is not nearly as strong in practice, but can be computed in linear
time.  It applies only to networks that are symmetric directed graphs and
hence can be modeled as undirected graphs, but, as remarked above, most real
world networks have this property.
It equals the minimum cover size for the variant of PDFL where we drop
the constraint that our paths be shortest paths, and only insist that they
be vertex-disjoint.  Let us call the task of finding such a cover as the
{\em unconstrained path-disjoint facility location problem} (UPDFL).
The algorithm for solving the UPDFL problem relies on three observations.

\begin{figure}
\begin{center}
\includegraphics[height=3in]{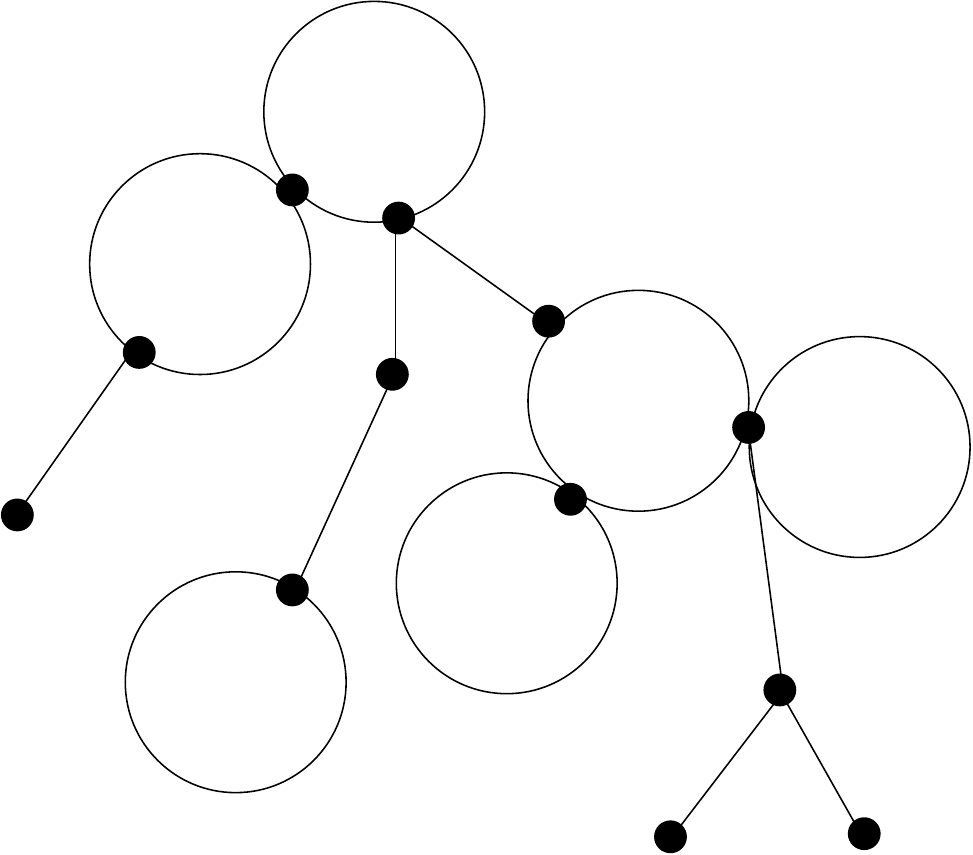}
\caption{\label{treefig}\small Decomposition into biconnected components.}
\end{center}
\end{figure}

\begin{enumerate}
\item Any undirected graph can be decomposed in linear time into its
2-connected components, the
resulting structure being treelike, with the 2-connected components playing
the role of vertices.  (See Figure \ref{treefig}.)
\item If our network is 2-connected, then any pair of facility locations in
it will cover all the customers it contains.
\item We already have an easily modifiable algorithm for PDFL on trees
(see Section \ref{sec:trees}).
\end{enumerate}

The first observation is attributed to Hopcroft in \citep{AHU}, which describes
the algorithm in detail.  To see the tree structure more clearly,
call a vertex an {\em articulation point} if deleting it from the graph breaks
the graph into two or more separate connected components, and consider the
graph  $t(G)$ derived from $G$ as follows.
The vertices of $t(G)$ are the leaves, the articulation points, and for each
2-connected component $C$, a new vertex $v_C$ representing the component.
The edges consist of all original edges that did not connect two vertices in
the same 2-connected component, together with edges $\{u,v_C\}$ for every
pair of an articulation point $v$ and a 2-connected component $C$ that contains it.
This graph cannot contain a cycle.
If it did, such a cycle would have to contain at least two vertices, each
representing either a vertex not in any 2-connected component or a maximal
2-connected set.  However, the existence of the cycle would imply that
the set consisting of all the vertices and 2-connected components represented
in the cycle would itself be 2-connected, a contradiction.
Thus $t(G)$ has no cycles, and is a connected graph since by assumption $G$ is.
Hence it is a tree.

The proofs of observations (2) and (3) go as follows:

\smallskip
\noindent
{\sc Proof of Observation 2}.
By definition,
if a graph $G$ is 2-connected and $u$ and $v$ are distinct vertices of $G$,
then $G$ contains two paths joining $u$ to $v$ whose only common vertices are $u$
and $v$ themselves.
Let $u$ and $v$ be two facility locations in our 2-connected graph $G$, and let $c$
be a customer vertex in $G$.
If $c \in \{u,v\}$ we are done, so we assume it is distinct from both.
Let $P_1(u)$ and $P_2(u)$ be two vertex-disjoint paths from $u$ to $c$,
and similarly, let $P_1(v)$ and $P_2(v)$ be two vertex-disjoint paths from $v$ to $c$.
Such paths exist by the definition of 2-connectivity.
We are done if either of $P_1(v)$ and $P_2(v)$ is vertex-disjoint from one
of $P_1(u)$ and $P_2(u)$.  So assume that both $P_1(v)$ and $P_2(v)$ intersect both
of $P_1(u)$ and $P_2(u)$.  At most one of them can contain the vertex $u$,
so assume without loss of generality that $P_1(v)$ does not contain $u$,
and let $w$ be the first vertex it contains that is an internal vertex of
one of $P_1(u)$ and $P_2(u)$, say, without loss of generality, $P_1(u)$.
Then the path that starts with $v$, proceeds along $P_1(v)$ until it hits $w$,
and then proceeds along $P_1(u)$ until it reaches $c$, is vertex-disjoint from
$P_2(u)$, and once again $u$ and $v$ cover $c$, and this holds for all cases. $\Box$

\smallskip
\noindent
{\sc Proof of Observation 3}.
Here is the needed modified version of the algorithm
of Section \ref{sec:trees}.
We assume there is at least one customer,
as otherwise the empty set is an optimal cover.

Call a 2-connected component a {\em leaf component} if it contains only one
articulation point.  As usual, call a vertex a {\em leaf} if it has degree 1.
Note that no leaf can be contained in a 2-connected component, and that all
nonleaf vertices that are not in a 2-connected components are articulation
points.
The {\em internal vertices} of a 2-connected component are all its vertices
that are not articulation points.
As in the algorithm of Section \ref{sec:trees}, we start with a pruning process.
If a vertex is a leaf in the current graph and not a customer, delete it to
get a new ``current graph.''  Similarly, if a leaf component has no internal
vertices that are customers, delete all its internal vertices, leaving only
its (single) articulation point.

After the pruning is complete, the final graph $G^*$ is such that every leaf is
a customer and every leaf component contains an internal customer.
There are two cases to consider:
\begin{enumerate}
\item $G^*$ contains no leaves or leaf components.
\item $G^*$ contains a leaf or leaf component.
\end{enumerate}

The first case implies that $G^*$ is a 2-connected graph.
By assumption, it must contain at least one customer.
If just one, say $v$, then $\{v\}$ is an optimal cover.
If more than one, then no single facility can cover all of them,
so pick any two customer vertices, say $u$ and $v$, and $\{u,v\}$ will be an optimal
cover by our second observation.

For the second case,
first note that since all the leaves are now customers,
they must be in any optimal cover.
Similarly, any cover must contain at least one facility in the interior
to each leaf component.
To see this, let $v$ be the articulation point for the component, and $c$
an interior point of the component that is a customer.
If our purported cover contains no facilities interior to the component,
there must be two paths to $c$ from facilities it does contain, but all
such paths must now contain $v \neq c$, and hence cannot be vertex-disjoint.
Let $F^*$ be a set consisting of all the leaves, together with one internal customer
from each remaining leaf component.
By the above argument, we know that the optimal cover must be at least of
size $|F^*|$.
We now argue that $F^*$ itself is a cover, and hence an optimal one.

Consider any customer vertex in $G^*$ that is not in $F^*$ and the corresponding
vertex $v$ in the underlying tree $t(G^*)$.
This is $v$ itself if $v$ is an articulation
point or a vertex not contained in any 2-connected component, and otherwise
is $v_C$ for the 2-connected component for which $v$ is an internal vertex.
Note that every leaf of this tree either is a member $F^*$ or corresponds
to a 2-connected component containing a member.
Thus if $v$ corresponds to a nonleaf of the tree, it must be covered
by $F^*$, since it has vertex-disjoint paths to at least two leaves.
The only other possibility is if $v$ is an internal vertex of a leaf component
$C$ of $G^*$.
But then $F^*$ must contain a vertex $f \neq v$ that is
one of $C$'s internal facility locations,
since it contains at least one and does not contain $v$.
Let $u$ be the single articulation point in $C$.
By our second observation, there are vertex-disjoint paths in $C$ from $v$ to
to $f$ and $u$.
Since $u$ itself must be connected to some leaf outside of $C$, this means
that $v$ has vertex-disjoint paths to two members of $F^*$.
Thus $F^*$ is an optimal cover. $\Box$

\begin{figure}[t]
\vspace{-.1in}
\begin{center}
\includegraphics[width=4in]{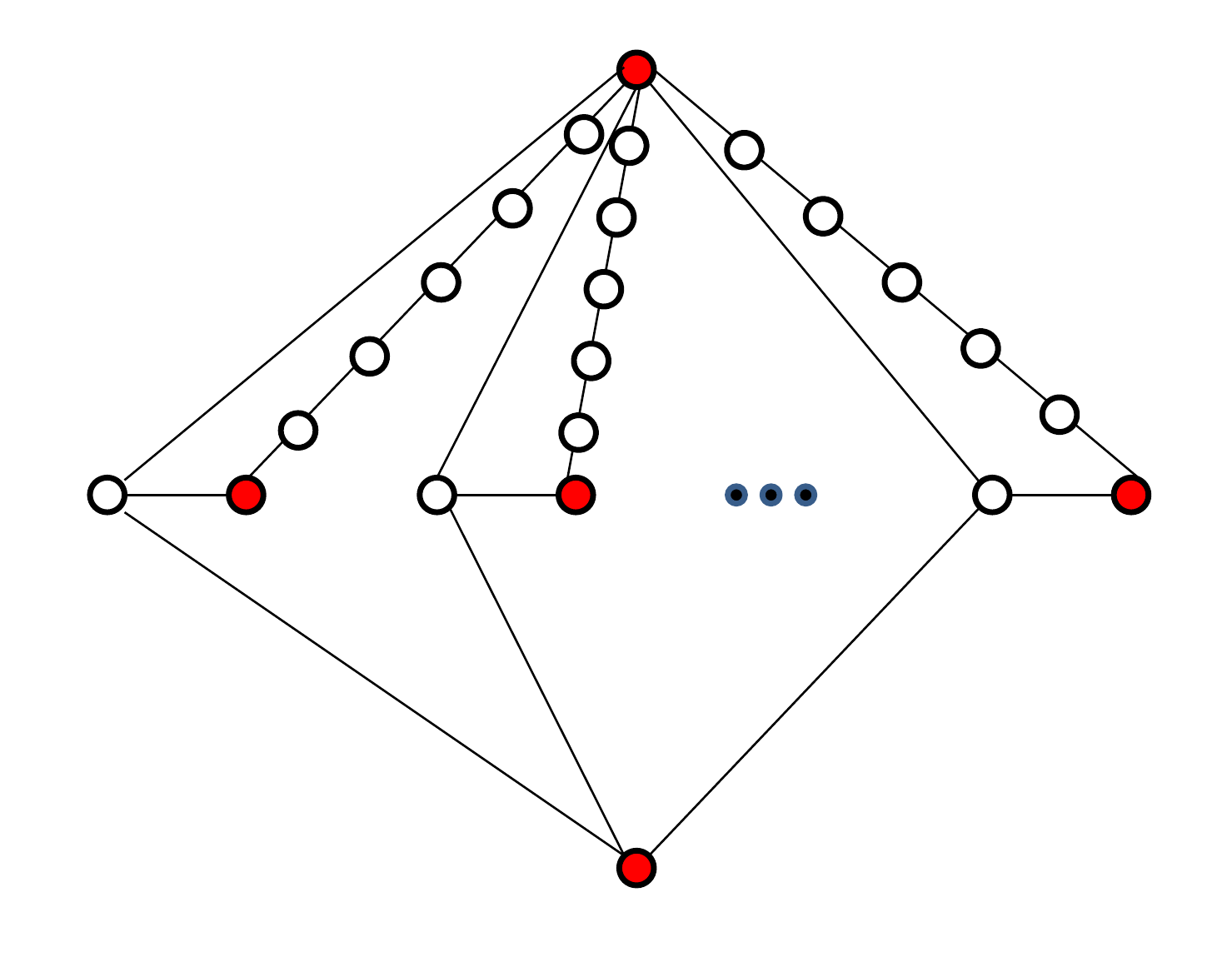}
\vspace{-.2in}
\caption{Instances where the UPDFL lower bound underestimates the optimal cover size by a factor that is linear in $n$.  The filled-in vertices are the only customers/facilities.}\label{UPDFLfig}
\end{center}
\vspace{-.3in}
\end{figure}

As with the Hitting Set Lower Bound for SDFL, the UPDFL lower bound for
PDFL is quite bad in the worst case.  For such examples we can restrict
ourselves to 2-connected graphs, since, as implied by the arguments above,
there is no advantage to allowing nonshortest paths except between vertices
in the same 2-connected component.  So consider the 2-connected graph schematized
in Figure \ref{UPDFLfig}, where the ellipsis stands for $N-3$ additional connected
vertex pairs in the middle row.
Here the filled-in vertices are the customers and the empty vertices are neither
customers nor facilities.  There are $N$ customers in the middle row, and
every PDFL cover must include all of them, since any shortest path
to a target customer in the row from any other customer
vertex (in the row above it or below it) has to go though the vertex connected to
the target on its left.  However, an optimal UPDFL cover consists simply of the two
customers at the top and bottom of the figure, since the length-6 path
from the top customer to each middle customer shares no vertices with
the length-2 path to it from the bottom customer.
The total number of vertices is $n = 7N+2$, and so the $N/2$ ratio between
the size of the two covers is linear in $n$.

Note that this is also a lower bound on how large a penalty we
can experience by limiting ourselves to shortest paths.
For any given instance, if there is a gap between the UPDFL lower bound and
the PDFL cover we construct, that gap is the sum of (1) the amount by
which the number of monitors
used in our constructed solution exceeds the optimal number,
plus (2) the penalty (in excess monitors needed)
for restricting ourselves to shortest paths.
(There may also, of course, be penalties for {\em not} using shortest paths,
such as added congestion and lessened monitoring accuracy.)

\section{Heuristics}\label{section:heuristics}
\subsection{Greedy Heuristic}\label{sec:greedy}
The simplest, and most general, heuristic that we tested is a straightforward
greedy heuristic, denoted by \greedy\ in what follows.
This is a heuristic for solving the problem described by the MIP in Section \ref{section:MIP},
and can be applied both to path-disjoint and set-disjoint versions of the problem.
It consists of three phases.
\begin{enumerate}
\item Construct the triples $\{(c,f,f'): \{f,f'\}$ covers $c$ in a path-disjoint
(set-disjoint) fashion\} and store them in appropriate data structures.
\item Initialize our cover $F' \subseteq F$ to the empty set.
Then, as long as $F'$ does not cover all $c\in C$,
find an $f \in F - F'$ such that $F' \cup \{f\}$ covers
a maximum number of additional elements.
Since all customers are facilities, there is guaranteed to be a facility that will
cover at least one more customer.
\item {\em Minimalize} the cover by testing
each facility in the cover in turn to see if its removal would still
leave a valid cover, and if so remove it.
\end{enumerate}
Note that, except for the fact that we allow a customer to cover itself,
this is essentially a greedy algorithm for the {\sc Set Cover by Pairs} problem.
It differs from the ``Greedy'' algorithm of \citet{HS05} for that problem, however.
The latter does not include a minimalization
phase, and in its growth phase it chooses the best {\em pair} to add.
It is hence likely to be significantly slower and less effective in practice than
our greedy variant.

Our implementation of {\greedy} uses randomized tie-breaking.
It also allows two options for starting the construction phase: either take
a pair that covers the maximum number of customers, or take a
random customer, which at least covers itself.
We also allow two options for the minimalization phase: either consider
facilities in the reverse of the order in which they were added to
the cover, or consider them in random order.

None of these options (or combinations of options) seems to dominate
any of the others, so here our primary tests cover
a multiple-run variant on {\greedy} ({\gr4}), which
performs 400 runs and returns the best solution found.
In those 400 runs, we cycle through
the four option combinations, using each for 100 runs.
Multiple runs help amortize the cost of triple generation.
The latter takes time $\Omega(|C|\cdot |F|\cdot |F|)$ and typically yields a
set $T$ of triples that is substantially smaller than the
$|C|\cdot |F|\cdot |F|/2$ upper bound, whereas the time for one
construction/minimalization cycle is $O(|T| + |F|^2)$.
Moreover, as we shall see, the actual overall running times are not large
enough to represent an obstacle in the proposed applications, where
the time to deploy a solution will greatly exceed the time to compute it.
Our experiments will, however, explore the tradeoffs between robustness
and number of iterations.

Aside from the triple-generation, which will be described in
Section \ref{section:triples}, the implementation of \greedy\ is fairly
straightforward, with two exceptions.
First, to avoid the multiple cache misses that would occur if we did many
random accesses into the list of triples $(c,f,f')$, we keep three copies
of that list, one sorted by the customer $c$, one by the larger-indexed
of $f$ and $f'$, and one by the smaller, with pointers into the start
of the portion for each choice of the sorting key.
The list indexed by $c$ is constructed as part of the triple generation
process.
The other two lists can be computed in linear time by first precomputing
how many triples there are for each key and reserving space for them
in the corresponding list, and then maintaining appropriate counters while
the lists are being constructed.
The relevant portions of the lists for each facility $f$ are accessed
at most four times, once at the start, once when (and if) $f$ is added
to the cover, once in preparation for minimalization, and once when
(and if) $f$ is deleted.
Another detail:  Because we are storing the data three times, we
trade a bit of time to save space, by running the triple generation code
twice.
The first run is used to compute the number of triples, which then allows
us to reserve the correct amount of space for each list of triples,
rather than simply pre-allocating room for the maximum
number $|C|\cdot|F|\cdot(|F|-1)$ of triples.

Our second nonobvious implementation detail consists of the data
structures used to handle minimalization, where we need an efficient
way of noticing when facilities are required, that is, can no longer
be deleted.
Let us call a pair of facilities $\{f,f'\}$ {\em valid} for a customer
$c$ if the pair covers $c$ in the appropriate fashion, and both its
members are still in our current cover.
A facility $f$ can become required for one of two reasons:
\begin{enumerate}
\item $f$ is also a customer and there are no longer any valid pairs for it.
\item There is a customer $c$ such that every valid pair for $c$ contains $f$.
\end{enumerate}
In order to determine when these situations occur, we keep track of
two types of counts.
First, for each customer $c$, we keep track of the current number
{\tt mycount}$(c)$ of valid pairs for $c$.
This makes it easy to check when the first of the above two conditions occurs.
Second, for each customer/facility pair $(c,f)$, we keep track of the current
number {\tt covercount}$(c,f)$ of valid pairs for $c$ that contain $f$.
This information is then stored in a set of linked lists, one for each
customer $c$ and number $n$ of pairs, where the list for $(c,n)$ contains
all those facilities $f$ with {\tt covercount}$(c,f) = n$.
To determine whether the second condition occurs, we simply look at the
lists for $(c,\mbox{\tt mycount}(c))$ for all customers with
{\tt mycount}$(c)$ less than the size of the current cover, since for a
cover of size $s$, there can be at most $s-1$ valid pairs containing the
same facility.
It is not difficult to see that the total time for maintaining the data
structures is $O(|T|)$, and the total time for spotting all the changes from
nonrequired to required is at most $|C|$ times the number of deletions, 
which is $O(|C|\cdot |F|)$.

\vspace{-.1in}
\subsection{Genetic Algorithm}

Genetic algorithms are variants on local search that
mimic the evolutionary process of survival of the fittest.
Our genetic algorithm, called \genetic\ in what follows,
uses the ``random key'' evolutionary strategy proposed by 
\citet{Bea94}, as modified by \citet{GonRes10a} so
as to bias the choices of parents and the operation of crossovers.
In this approach, the ``chromosomes'' that do the evolving are not
solutions themselves, but ``random-key'' vectors from which solutions
can be derived.
(This is in contrast to traditional genetic algorithms, in which the chromosomes
are the solutions themselves.)

As with \greedy\, which it uses as a subroutine,
our \genetic\ algorithm applies to the MIP formulation
of the set-disjoint and path-disjoint versions of the problem, and does not
directly exploit any graph-theoretical properties of the underlying network.
Let the set of facilities $F = \{f_1,f_2,\ldots ,f_k\}$.
In \genetic, each chromosome is a 0-1 vector
$( \mathit{gene}_1, \ldots, \mathit{gene}_k)$ of length $k$.
We derive a solution (cover) $F'$ from such a chromosome as follows:
Start with $F' = \{f_i: \mathit{gene}_i = 1,\ 1 \leq i \leq k\}$.
If $F'$ is already a cover, halt.
Otherwise, complete $F'$ to a cover using our basic \greedy\ algorithm
(without randomization or minimalization).
The ``fitness'' of the chromosome is then the size of the resulting
cover $F'$.

The overall genetic algorithm starts by creating a population of $p$
randomly generated chromosomes, in which each gene has an equal probability
of being 0 or 1.
The population then evolves in a sequence of generations.
In each generation we start by computing the solution for each member
of the population, yielding its fitness.
The top 15\% of the population by fitness (the ``elite'' members)
automatically pass on to the next
generation's population.
In addition, to provide diversity, 10\% of the next generation consists
of randomly generated chromosomes, like those generated initially.
For the remaining 75\%, we repeat the following ``biased crossover'' construction:
Pick a random member $(x_1,x_2,\ldots,x_k)$ of the top 15\% of the current
generation's population and a random member $(y_1, y_2,\ldots,y_k)$ of
the remaining 85\%.
The ``child'' $(z_1,z_2,\ldots,z_k)$ of this pairing is determined as
follows:
For each $i$, independently, set $z_i = x_i$ with probability 70\%,
and otherwise set $z_i = y_i$.

This scheme insures that the best solution always survives into
the next generation, where it may continue as champion or be dethroned
by a better solution.
Generations are continued until $q$ have gone by without any
improvement in the fitness (cardinality) of the best solution, which
we then output.
Our code includes the option of minimalizing this final solution, but
in practice this never has helped. 

In our results for \genetic\, we typically took $p = \min\{300,|F|\}$
and $q = |V|$.
For these and the other parameters of the algorithm (top 15\%, etc.)
we chose values based on preliminary experimentation and on intuition derived from
applications of this approach to other problems.
They appear to yield a good tradeoff between running time and
quality of solution for our application, but we make no claim as to
their optimality.

In our experiments, \genetic\ did not prove competitive for the set-disjoint
case.
The much faster heuristics of the next two sections,
which exploit the properties of that case, found optimal solutions for
all our test instances.
In the path-disjoint case, where, for our larger test instances,
multiple-iteration \greedy\ is the only competition, and where we have no
readily-computable lower bound that is as good as {\lb} is for the
set-disjoint case, \genetic\ can at least give us some idea of
the extent to which \gr4\ can be improved upon.

\subsection{The Single Hitting Set Heuristic (\shs)}

This section describes a simple heuristic that applies to the set-disjoint variant
of our problem (SDFL), and is based on Lemma \ref{neighborlemm}
and our Hitting Set lower bound for that variant.
Instead of optimally solving the NP-hard Hitting Set problem described
in Section \ref{section:HSLB},
we find a not-necessarily optimal solution using a straightforward greedy
hitting set heuristic:
Let $X \subseteq F$ denote the hitting set we are building, initially the empty set.
While $X$ is not a valid hitting set, repeatedly add that facility that hits
the most as-yet-unhit (customer,neighbor) pairs, ties broken randomly,
where a facility $f$ hits a pair $(c,v)$ if either $f=c$ or $v \notin N(c,f)$.
Once $X$ is a valid hitting set, we check whether it is also
a valid cover of all the customers and, if not, we extend it to one
using \greedy, this time including randomization and minimalization.
(Minimalization is applied even if $X$ is itself already a cover.)

As in \gr4\, we use this simple heuristic in a multi-iteration scheme,
where we take the best solution found.
This allows us to exploit the variance introduced
by all the randomization -- in addition to all the randomized tie-breaking,
we also alternate between reverse
delete and randomized delete in the minimalization phases.
We shall denote this scheme by \shs$(k)$, where $k$ is the number of
iterations, again typically 400.

Since \shs$(k)$ is a polynomial-time algorithm, Theorem \ref{equivalence} implies
that we cannot prove a polylogarithmic worst-case performance guarantee for it
unless a widely-believed conjecture is false.
However, in practice the initial hitting set $X$ is typically close to a cover,
and so the following provides an informative bound.

\begin{theorem}  \label{SHSbound}
For a given instance $I$ of SDFL based on a graph $G=(V,A)$, let {\sc OPT}$(I)$ be
an optimal cover, and suppose we run {\shs} one time, with
$X(I)$ being the hitting set generated and \shs$(I)$ being the final cover.
Then, if we set $\Delta = |\mbox{\shs}(I)|-|X(I)|$, the number of
facilities that needed to be added (removed) to turn $X$ into a minimalized cover, we have
$$
\mbox{\shs}(I) \leq (\ln(|A|)+1)\cdot \mbox{\sc OPT}(I) + \Delta.
$$
\end{theorem}

\begin{proof}
This follows from the fact that the optimal solution to the hitting set problem
of Section \ref{section:HSLB} is a lower bound on OPT$(I)$,
together with the result of \citet{Johnson74} and \citet{Lovasz75} that the greedy algorithm
for {\sc Set Cover} (and hence for {\sc Hitting Set}) has worst-case ratio
at most $\ln n + 1$, where $n$ is the size of the set being covered/hit.
\end{proof}

Note that the $\ln n +1$ bound for the greedy {\sc Hitting Set} heuristic
is essentially the best possible guarantee for
a polynomial-time algorithm for this problem, as follows from 
\citet{BGLR93},
\citet{Feige98}, and
\citet{AMS06}.
In practice, however, the greedy {\sc Hitting Set} heuristic typically
returns a set of size much closer than this to the optimal, often only a
single-digit percentage above optimal.
As we shall see, the performance of \shs$(k)$ is similarly good.

\subsection{The Double Hitting Set Heuristic (\dhs)}

This section describes a second heuristic for SDFL that exploits Lemma \ref{neighborlemm},
as well as the way OSPF weights tend to be set.
It involves computing two hitting sets rather than one, but is a good
match for {\tt SHS}, in that both heuristics perform very well and neither
consistently dominates the other.
 
As remarked above, while OSPF can split traffic in the
case of multiple shortest $c\rightarrow f$ paths, the amount of splitting
that occurs in practice seems to be limited.  
The following two definitions provide a new way to quantify this effect.
\begin{Def}
A potential facility location $f$
is {\em good for} customer vertex $c$ 
if $f\ne c$ and $|N(c,f)| = 1$, that is, all shortest $c\rightarrow f$
paths leave $c$ via the same arc.
\end{Def}

If there is just one shortest $c\rightarrow f$ path, clearly $f$ is good for $c$, 
but $f$ may be good for $c$ even when there are many $c\rightarrow f$ paths.
Note that if $f$ is good for $c$, then, under OSPF, there can be no splitting at
$c$ of the traffic from $c$ to $f$, although splitting at later vertices in
the path would be possible.

\begin{Def}
Let $t$ be an integer, $1 \leq t \leq |F|$.
A customer vertex $c$ is {\em $t$-good} if there are
$t$ or more good potential facility locations $f$ for $c$ (and {\em $t$-bad}
otherwise).
\end{Def}

For our test instances, inspired or derived from real-world networks and
models, a large majority of the customers tend to be $t$-good, even for $t\sim |F|/2$,
and almost all are $1$-good.
Our Double Hitting Set heuristic (DHS) takes $t$ as a parameter, and is designed to
cover the $t$-good customers nearly optimally, via a union $X\cup Y$ of two hitting
sets, leaving the $t$-bad customers, if any, to be covered via the \greedy\ algorithm.
Let $C_t$ denote the set of all $t$-good customers.
For a given $t$, we shall, when it makes a difference, denote the algorithm by ${\rm DHS}_t$.

The first step of ${\rm DHS}_t$ is
to choose one good potential facility location $f \in F$ for each $c \in C_t$,
without choosing too many vertices in total.
For each such $c$, let $S_c$ be the set consisting of $c$ plus all
the good potential facility locations for $c$.
By the definition of $t$-good and the fact that $C \subseteq F$, we must have
$|S_c| \geq t+1 \geq 1$.
We want to choose a small set $X \subseteq F$ such that $X\cap S_c\ne \emptyset$
for all $t$-good customers $c$.
In other words, we want an $X$ that hits (intersects) all such sets $S_c$.
We again settle for using the standard
greedy {\sc Hitting Set} heuristic, as described in the previous section,
to construct a hitting set $X$ that we hope is near-optimal.

Next, let $C' \subseteq C_t$ be the set of $t$-good customers $c$ covered by $X$, either
because $c \in X$ or because there exist distinct $f_1,f_2 \in X$ such that
$c$ is covered by $\{f_1,f_2\}$ in a setwise disjoint fashion.
Let $C'' = C_t - C'$, and for each customer $c \in C''$,
choose a facility $f_c \in X\cap S_c$.  Note that by definition of $C''$,
we have $f_c \neq c$, so let $x_c \in N(c)$ be the unique neighbor of $c$
which is on all shortest paths from $c$ to $f_c$.  This choice is unique since
the fact that $f_c \in S_c$ implies that $f_c$ is good for $c$, and so
by definition all shortest paths from $c$ to $f$ leave $c$ through the same arc.

Our second {\sc Hitting Set} problem is designed so any solution $Y$ will, for
each $c \in C''$, contain either $c$ itself or a facility $f \neq c$ such that
$\{f_c,f\}$ covers $c$ is a setwise disjoint fashion.
The {\sc Hitting Set} problem is specified as follows.
For each such customer vertex $c$, let $F_c$ consist of $c$
together with all potential facility
locations $f\ne c$ for which all shortest $c\rightarrow f$ paths
{\em avoid} $x_c$.
Note that by Lemma \ref{neighborlemm}, for any $f \in F_c$
(other than $c$ itself), the fact that $f_c$ is good for $c$
and the definition of $x_c$ together imply that the pair $\{f_c,f\}$ will cover $c$
in the desired fashion.
We thus have the following:

\begin{lemma} \label{XcupY}
Given $X$ as above, any hitting set $Y$ of 
$\{F_c: c \in C''\}$ is such that $X\cup Y$ is a cover of $C' \cup C'' = C_t$,
the set of all the $t$-good customer nodes $c$.   
\end{lemma}

\noindent
For the DHS algorithm, we also use the greedy {\sc Hitting Set} heuristic
to produce $Y$.  At this point, we have a cover $X \cup Y$ for $C_t$.
DHS concludes by using \greedy\ to extend $X \cup Y$ to a minimalized
cover $F'$ of {\em all} the customers, not just the $t$-good ones.

As with {\shs}, we will consider multi-iteration schemes, \dhs$_t(k)$,
where we perform $k$ independent runs of \dhs$_t$, in each using randomized
tie breaking when running the greedy {\sc Hitting Set} heuristic
and the final \greedy\ completion phase, and alternating between
reverse and random deletes in the minimalization phases.
Our main experiments will be with \dhs$_1(400)$ and
\dhs$_{\lfloor|F|/2\rfloor}(400)$, which seemed somewhat complementary
in our initial, single-iteration trials.
For simplicity, we shall denote the latter by {\dhs}$_H$(400)
(with the $H$ standing for ``half'').

As with SHS$(k)$, since DHS$_t(k)$ is a polynomial-time algorithm,
Theorem \ref{equivalence} implies
that we cannot prove a polylogarithmic worst-case performance guarantee for it unless
a widely-believed conjecture is false.
However, in practice the initial hitting set $X$ is typically quite small,
and often there are very few $t$-bad customers, so meaningful bounds are
possible.

\begin{theorem}  \label{DHSbound}
Suppose $1 \leq t \leq |F|$, and let $I$ be an instance of SDFL
with at least one $t$-good customer.
Let ${\rm OPT}_t(I)$ be the size of an optimal
cover for the set $C_t$ of all $t$-good customers in $I$.
Suppose we run ${\rm DHS}_t$ one time on $I$, with $X(I)$ and $Y(I)$
being the two hitting sets generated, which, as observed constitute a cover
of $C_t$.  Then,
\begin{eqnarray*}
|X \cup Y| \leq (\ln(|C_t|)+1) \cdot{\rm OPT}_t(I) + |X|
\end{eqnarray*}
and, if $t \geq |F|/3$, we have
\begin{eqnarray*}
|X \cup Y| \leq (3.47\cdot\ln(|C_t|)+2)\cdot{\rm OPT}_t(I)
\end{eqnarray*}
\end{theorem}

\noindent
Note that, for our implementation choice of $t = \lfloor |F|/2 \rfloor$, we have
$t \geq |F|/3$ as long as $|F| \geq 2$.
Also, the cover of all of $C$ produced by ${\rm DHS}_t$ can exceed $|X \cup Y|$ by at most
the number $|C-C_t|$ of $t$-bad customers, so if
almost all customers are $t$-good, the above bounds will reflect the overall
performance of DHS.
Finally, note that the logarithm here is of $|C|$ or something less, whereas for
the analogous Theorem \ref{SHSbound} for SHS, the logarithm was over the potentially
much larger $|A|$, which might affect the relative constant factors in the bounds.

\smallskip
\begin{proof}
Observe that the sets $F_c$ in our second hitting set constitute a subset of those
used in our {\sc Hitting Set} lower bound, restricted to the $t$-good customers.
Thus the optimal hitting set for this second
problem cannot be larger than ${\rm OPT}_t(I)$.
Thus, by the results of \citet{Johnson74} and \citet{Lovasz75} 
and the fact that, by assumption,
$C_t \neq \emptyset$, we have
$|Y| \leq (1 + \ln |C|)\cdot{\rm OPT}_t(I)$, and the first claim follows.

We cannot apply the same argument to the set $X$, since in constructing
$X$ we restricted attention to good potential facility locations, and ${\rm OPT}_t(I)$
need not be so restricted.
However, a different argument applies when $t \geq |F|/3$,
which will lead to the proof of our second claim.
This choice of $t$ means that, for each $t$-good vertex $c$, $|S_c| \geq |F|/3$.
Therefore, $\sum_{\mbox{$c$ is $t$-good}}|S_c| \geq |C_t|\cdot|F|/3$,
and so some potential facility location $f$ must be in at least $|C_t|/3$
of the sets.
Consequently, the greedy choice must hit at least that many sets.
By essentially the same argument, it must hit at least $1/3$ of
the remaining unhit sets each time it makes a choice, and so
must have completed constructing its hitting set $X$ after $y$ steps, where
$(2/3)^y|C| < 1$.
This means it suffices that $y > \ln(|C|)/\ln(3/2) > 2.47\cdot\ln(|C|)$,
and so $|X| \leq 1+ 2.47\cdot\ln(|C|)$.
The second claim follows.
\end{proof}

\subsection{Triple Generation.}\label{section:triples}
All our cover-generation algorithms require the generation of the
set of relevant triples $(c,f_1,f_2)$ (set-disjoint or path-disjoint)
used by our basic \greedy\ algorithm -- the algorithms other than \greedy,
either use it as a core-subroutine (\genetic) or as a fallback device to
complete the cover they generate (SHS and DHS).
Because this set is potentially of size $\Theta(|C|\cdot|F|^2)$,
and would naively take significantly more time than that to generate,
its computation typically constitutes a significant portion of the running time.
In this section we show how the time can be reduced to $O(|V|^3)$ for real-world
networks, in both the set- and path-disjoint cases.

In what follows, let $n = |V|$ and $m=|A|$.
For real-world data networks, we may assume that $m \leq an$
for some relatively small constant $a$.
We also have that $|C| \leq |F| \leq n$, although both $|C|$ and $|F|$
may be proportional to $n$.
Let $T_S$  and $T_P$ be the sets of triples needed by \greedy\ in
the set-disjoint and path-disjoint cases, respectively.

Let us first consider the set-disjoint case.
A triple $(c,f_1,f_2)$ is in $T_S$ if $c\in C$,
$f_1,f_2 \in F$, and no shortest path from $c$ to $f_1$ shares any vertex
other than $c$ with any shortest path from $c$ to $f_2$.  The naive way
to test this would be to construct, for each $f_i$, the set $S_i$ of vertices
on shortest paths from $c$ to $f_i$, and then testing whether $S_i \cap S_j
= \{c\}$.  These sets could conceivably be of size
proportional to $n$, yielding a running time that could be proportional
to $|C||F|^2n$, so the best we can say about the running time of such an
algorithm is that it is $O(n^4)$.

However, recall from Lemma \ref{neighborlemm} that if $S_i \cap S_j$ contains
some vertex other than $c$, then it contains a vertex that is a neighbor
of $c$.  Thus we may restrict attention to the sets $N_i$ of vertices
adjacent to $c$ that are on shortest paths to $f_i$.
To compute these sets we first construct a shortest path graph from
$c$ to all other vertices in time $O(m\log n)$, and then for each
facility $f$ we can compute $N_i$ by a single backward traversal of the shortest
path graph from $f$ in time $O(m)$.
For each pair $f_i,f_j$, the intersection test can
then be performed in time $O(outdegree(c))$.
Amortized over all customers in $C$, the tests for a given pair $f_i,f_j$
take time $O(\sum_{c \in C} outdegree(c)) = O(m)$.
Thus we get an overall time bound of $O(m \log n|C| + m|F|\cdot|C| + m|F|^2)
= O(n^2m)$ under our assumptions that $m$ is $O(n)$.

For the pathwise-disjoint problem, the naive algorithm is even worse,
since for a given triple $c,f_1,f_2$, there may be exponentially many
paths of potential length $\Omega(n)$ to compare.  Here we must be
more clever.  We first observe that we can actually reduce the
test to a simple network flow problem, by augmenting the shortest path
graph with a new vertex $t$, adding
arcs from $f_1$ and $f_2$ to $t$, giving all vertices and edges
capacity one, and asking whether there is a flow
of value 2 from $c$ to $t$.  Next, we observe that this is a particularly
easy flow problem, which can be solved in linear time as follows.

As before, for a given customer $c$ we begin by computing a shortest path
graph from $c$ to all other vertices in time $O(m\log n)$, to be amortized
over all facilities.
Set the capacities of all the arcs to 1.
We then add a new vertex $t$, and edges of length and capacity 1 from
all facilities $f$ to $t$.
To handle a given candidate pair $(f_1,f_2)$, we pick some shortest
path from $c$ to $f_1$ using the shortest path graph,
and extend it to a shortest path from $c$ to $t$ by adding the link from
$f_1$ to $t$.
This yields a flow of 1 from $c$ to $t$.
Now add a link of capacity 1 from $f_2$ to $t$.
There will be disjoint paths from $c$ to $f_1$ and to $f_2$ if and only
if there is now a flow of value 2 from $c$ to $t$,
and such a flow will exist if and only if
there is an augmenting path from $c$ to $t$ in the graph.
This can be found in linear time by a breadth-first search in the
residual graph, where each arc $(u,v)$ in the first path is replaced by
an arc $(v,u)$ of capacity 1.

Finally, we observe that, having fixed the shortest path to $t$ through $f_1$
and constructed the residual graph, we can actually solve the network
flow problems for all pairs $(f_1,f_i)$ in parallel with a single
breadth-first search.
The overall running time thus becomes
$O(m\log n|C| + m|C|\cdot |F|) = O(n^2m)$, once
again $O(n^3)$, assuming $m=O(n)$.

What we have just described is actually the algorithm for generating
path-disjoint triples when ``disjoint'' only means ``arc-disjoint.''
To get vertex-disjointness as well, one replaces each vertex $v$ in
the shortest path graph be a directed arc $(v_{in},v_{out})$, with
each arc $(u,v)$ replaced by $(u,v_{in})$ and each arc $(v,u)$
replaced by $(v_{out},u)$.

In conclusion, we note that for our {\lb} code in the set-disjoint case,
some of the above work can be avoided, as the triples need only be
considered implicitly.
To construct the {\sc Hitting Set} linear program, we need only consider
(customer,neighbor,facility) triples, at an amortized cost 
(and space) of $O(mn)$.
The same triples can be reused in the verification phase, where we
check to see whether the computed hitting set is actually a feasible solution
to our problem, although now $O(mn^2)$ time may be required.

\section{Network Test Instances}\label{sec:instances}\label{section:instances}

We tested our algorithms on both synthetic and real-world instances.
The two classes modeled different types of data networks and had
distinctly different characteristics, enabling us to test the generality
of our methods.
In this section we will discuss each type of instance,
and compare the structural properties of the resulting testbeds.

\subsection{Synthetic LAN/WAN Instances}
Our 630 synthetic instances were designed to reflect the structure of large
real-world local- and wide-area networks (LAN's and WAN's) and were sized
so that we could study
the scalability of our algorithms and the solutions they produce.
They were generated via the following four-step process.
\begin{enumerate}
\item
A transit-stub skeleton graph is generated using the
Georgia Tech Internetwork Topology Models (GT-ITM) package
\citep{CalDoaZeg97a}.  We generated 10 graphs each for parameter
settings that yielded $|V|=50,100,190,220,250,300,558,984$.
(The value of $|V|$ is an indirect result of one's choice of the
allowed input parameters.)
\item
Traffic demand is generated between all pairs of vertices in the
skeleton graph using a gravitational model (described below) with the shortest
path metric.
\item
Using the heuristic network design algorithms of \citet{BurResTho07a},
we determine link capacities and OSPF link weights such that
all traffic can be routed in a balanced way on the resulting network
using the OSPF routing protocol.
Links with no traffic routed on them are deleted.
\item
For specified numbers of customer and facility location vertices, 
the sets $C \subseteq F$ are then randomly generated.
\end{enumerate}

\vspace{-.1in}

\subsubsection{Step 1: Transit-Stub Skeleton Graphs}
Transit-stub graphs \citep{CalDoaZeg97a}
are hierarchical graphs made up of transit vertex
components and stub vertex components.
The stub node components can be thought of as access networks while
the transit node components make up the backbone of the network.
See Figure \ref{transitstubfig}.
\begin{figure}[t]
\begin{center}
\includegraphics[width=11cm,angle=0]{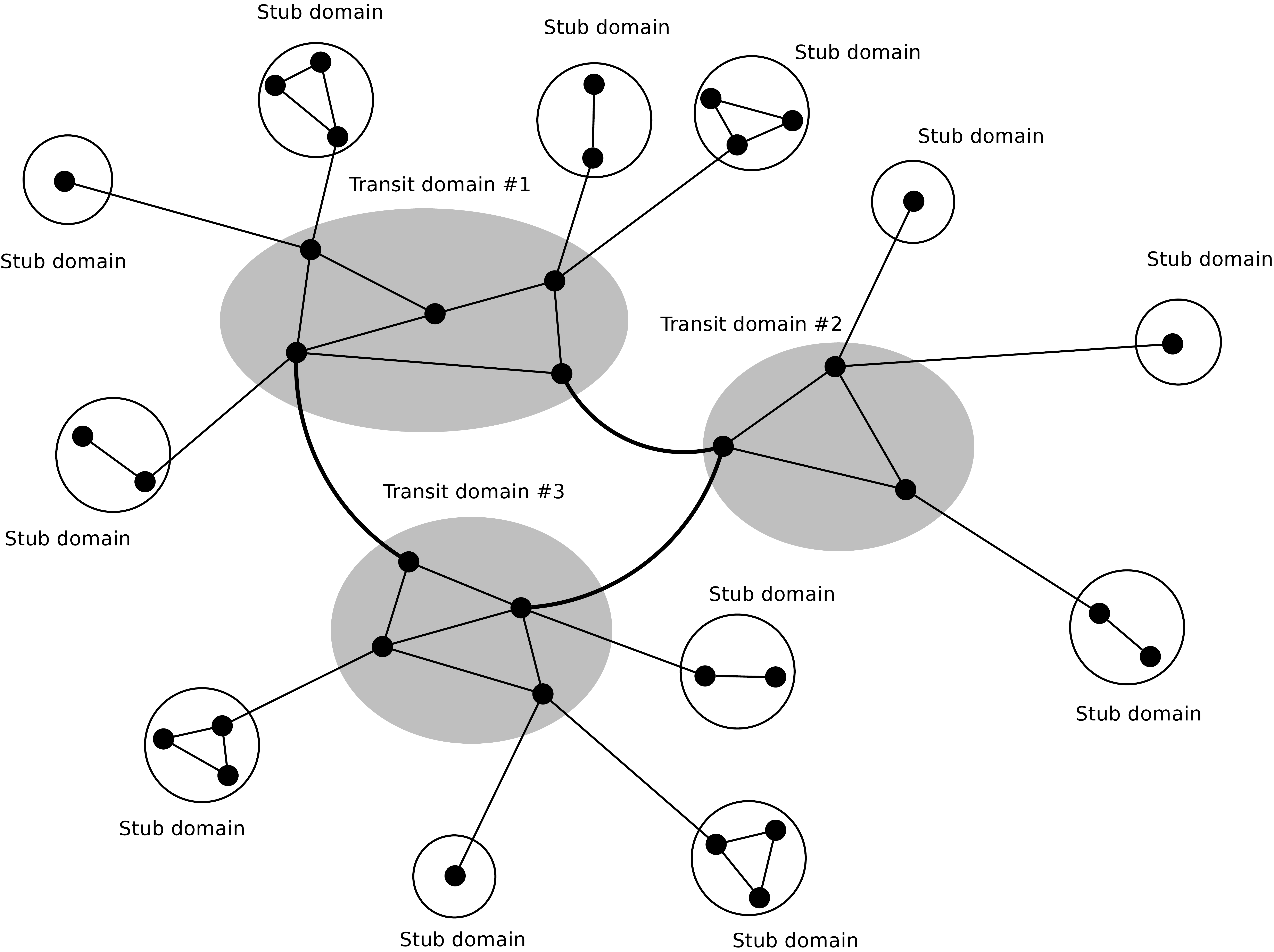}
\caption{A transit stub network with three transit domains 
and 16 stub domains.}
\label{transitstubfig}
\end{center}
\vspace{-.25in}
\end{figure}

The GT-ITM package provides a parameterized method for generating
such graphs randomly.  
The relevant papers and documentation are not totally consistent or clear, and
the following description is our best guess as to how the networks are
generated, based on the documentation and our inspection of selected outputs.
The generation process involves many parameters, but we based our
generation on the settings used in the illustrative example identified
as ``file t100'' in the documentation file {\tt itm-doc.txt} available at
{\tt http://www.cc.gatech.edu/projects/gtitm/}, varying only
the parameters $T$, $N_T$, $S$, $N_S$, which determine the number of vertices.
 
The constructed graphs consist of $T\cdot N_T$ {\em transit vertices}
assigned randomly to $T$ {\em transit domains}.
It is required that $T \geq 1$ and $N_T \geq 1$, and it appears that
the generation process starts by placing one transit vertex in each
transit domain so that none are empty.
In addition, there are $S\cdot T\cdot N_T$ {\em stub domains} (one per transit vertex),
and $S\cdot N_S\cdot T\cdot N_T$ {\em stub vertices},
assigned randomly to the stub domains.
Again we must have $S \geq 1$ and $N_S \geq 1$ and it appears that the
generation process starts by placing one stub vertex in each stub domain.
Edges appear to be added to the network as follows:
\begin{itemize}
\item
Each transit domain has an average of two edges to other transit domains.
The endpoints of each edge are chosen randomly from the transit vertices in
the two relevant domains.
If the chosen edges fail to connect all the transit domains, they are deleted
and the experiment is repeated until we do get a connected graph.
\item
Each transit vertex is associated with an average of $S$ stub domains,
and is connected by an edge to a random member of each.
\item
For each pair of vertices within a given transit domain, there is an edge
with probability 0.6.
If the resulting choices do not make the transit domain connected, the experiment
is repeated until a connected graph results.
\item
For each pair of vertices within a given stub domain, there is an edge
with probability 0.42.
If the resulting choices do not make the stub domain connected, the experiment
is repeated until a connected graph results.
\end{itemize}
In addition, it appears that the code
also adds a few extra edges between stub domains or between stub and transit domains,
even though there should be none, given the parameter settings we chose.
We deduce this from the fact that, although the above description suggests that the
largest 2-connected component should be of size $\max(N_S,T \cdot N_T)$, we typically find
that the largest is two to four times larger than that.

Note that this process yields a network with a fixed number $T\cdot N_T\cdot S\cdot N_S$ vertices,
but a random number of edges (some of which may be deleted in Step 3).
Before proceeding to Step 2, we convert our graphs to symmetric directed graphs by replacing
each edge $\{u,v\}$ with the two arcs $(u,v)$ and $(v,u)$.
Since the original graph was designed to be connected, this directed graph will be
strongly connected.
The GT-ITM generator also produces a length $l(u,v)$ for each edge $\{u,v\}$, which we
use in Step 3.


We generated 10 graphs each for eight different sets of parameters, chosen
to yield increasing values for $|V|$, from 50 to 984.
In Table \ref{transitstubparams} we present the parameter choices, the corresponding
values for $|V|$, the final average number of edges and the average degree
(when the graph is viewed as undirected).
The graphs for instances with up to 558 vertices were generated by the 32-bit
version of the GT-ITM code, which is apparently no longer supported.  Those for
the graphs with 984 vertices, which we generated more recently, were produced using
the new, 64-bit code from Version 2.33 on {\tt ns-allinone},
downloaded from {\tt http://sourceforge.net/projects/nsnam/files/}.
This version apparently has a bug, in that the parameters for the edge densities seem to
be ignored.
Consequently, both the stub and transit domains start out as complete graphs,
and their sparsification, while still substantial, is due only to the Step 3 deletions,
described in more detail below.

\begin{table}
\begin{center}
{
\begin{tabular}{rrrr|rrc}
$T$ & $N_T$ & $S$ & $N_S$ & $|V|$ & Average $|A|$ & Average Degree\\ \hline
1 & 2 & 3 & 8 & 50 & 91.1 & 3.6\\
1 & 4 & 3 & 8 & 100 & 177.2 & 3.5\\
2 & 5 & 3 & 6 & 190 & 284.6 & 3.0\\
2 & 5 & 3 & 7 & 220 & 360.2 & 3.3\\
2 & 5 & 3 & 8 & 250 & 435.0 & 3.5\\
2 & 6 & 3 & 8 & 300 & 531.6 & 3.5\\
3 & 6 & 3 & 8 & 558 & 1172.2 & 4.2\\
4 & 6 & 4 & 10 & 984 & 2083.3 & 4.2\\ \hline
\end{tabular}
}
\caption{Input parameters and measured statistics of generated transit-stub graphs.}\label{transitstubparams}
\end{center}
\vspace{-.25in}
\end{table}

\subsubsection{Step 2: Traffic Demands}
The traffic demands are created via a randomized ``gravitational'' method from \citet{FT00}.
We first generate random numbers $o(v)$ and $d(v)$ for each vertex $v$,
chosen independently and uniformly from the interval $(0,1)$.
Here ``o'' stands for origin and ``d'' for destination.
Then, for each pair of vertices $(u,v)$, we compute $D(u,v)$,
the shortest length (in edges) of a path from $u$ to $v$.
Let $Dmax$ be the largest of these distances.
Then, for any ordered pair $(u,v)$ of distinct vertices, we choose a random
number $r(u,v) \in (0,1)$ and set the traffic demand from $u$ to $v$ to be
the (nonzero) value
$$
{\Large t(u,v) = r(u,v) \cdot o(u) \cdot d(v)\cdot e^{ -\left(\frac{D(u,v)}{Dmax}\right)}}.
$$
Note that, since our network is strongly connected, all of this traffic can
be routed.

\subsubsection{Step 3: OSPF Weights}
We next choose OSPF weights and arc capacities that allow for an
efficient routing of the traffic
specified by the matrix $t(u,v): u,v \in V$.
To do this, we apply the survivable IP network design tool developed by Buriol
et al.\ and described in \citet{BurResTho07a}.
The tool takes as input the skeleton graph and the traffic matrix.
For each arc $(u,v)$, it produces an integer OSPF weight $w(u,v)$
in the range $[0,30]$ and an integer capacity $c(u,v)$,
such that if all traffic is routed according
to the OSPF weights (with traffic splitting if there is more
than one shortest path from one vertex to another),
no arc will have traffic load exceeding 85\% of
its capacity.

It does this while attempting to minimize the total arc cost
$\sum_{(u,v)}l(u,v)\cdot c(u,v)$.
The optimization is done via a biased random key genetic algorithm,
using a population of size $50$ and run for $100$ iterations.
For more details, see \citet{BurResTho07a}.
(Running with larger parameter values would likely produce lower-cost solutions
to the design problem, but this would increase the running time,
and our goal here is simply to produce a sparse network on
which the demand can be routed.)

Note that any arc $(u,v)$ that is on some shortest path for the computed OSPF
weights must have nonzero capacity $c(u,v)$, since every traffic demand
is nonzero and so every shortest path must carry {\em some} traffic.
We thus can safely delete all capacity-zero arcs from our network.
In order to insure that our graphs remain symmetric, however, we only delete
a capacity-zero arc $(a,b)$ if its partner $(b,a)$ also has zero capacity.

\subsubsection{Step 4: Customers and Potential Facility Locations}

For each of our 80 skeleton graphs (10 for each of our eight sets of GT-ITB parameters),
we generated seven networks, differing in
their choices of the sets $C \subseteq F$ of customer vertices and
potential facility locations.
The sizes were chosen as a set of fixed fractions of the set $V$ of vertices.
Let (C$x$,F$y$) denote the set of instances with
$|C| = \lceil|V|/x\rceil$ and $|F| = \lceil|V|/y\rceil$.
For each graph we generated one instance of each of the following types:
(C1,F1), (C2,F2), (C4,F4), (C8,F8), (C2,F1), (C4,F1),
and (C8,F1). 

Our synthetic instance testbed thus contains a total of 560 instances,
each consisting of a network, OSPF weights (which determine the shortest
paths in the network), and the sets $C$ and $F$.
The full set of synthetic instances is available from author
Resende, and currently can be downloaded from {\tt http://mauricio.resende.info/covering-by-pairs/}.

\subsection{Real-World ISP-Based Instances}
The 70 real-world instances in our testbed were derived from 10
proprietary Tier-1 Internet Service Provider (ISP) backbone networks and many
of them use actual OSPF weights.
Because of the proprietary nature of these instances, we will not be able to report as detailed
and precise results about them, although we shall summarize how they differ from
the synthetic instances, and how the results for
the two instance types compare.

The ISP networks ranged in size from a little more than 100 routers to
over 5000, each with between $3.5|V|$ and $4.3|V|$ arcs
(similar to the range for our synthetic instances).
We shall denote them by
R100a, R100b, R200, R300, R500, R700, R1000, R1400, R3000, and R5500.
where the number following the R is the number of routers, rounded
to the nearest multiple of 100.
(Because of their proprietary nature, these instances cannot be made publicly
available.)

For four of these instances (R100a, R100b, R200, and R500), we only had information
about the topology, but not the roles of the routers.
For each of these, we constructed 16 instances,
starting with the case in which $F = C = V$.
The other instances also had $F = V$, but $C$ was a random sample
of roughly $1/2$, $1/4$, or $1/8$ of the vertices.
Following our notation for the synthetic instances, we shall refer
to these four alternatives as (C1,F1), (C2,F1), (C4,F1), and (C8,F1).
In addition to the single (C1,F1) instance, we generated five distinct instances
for each alternative from {(C1,F2), (C1,F4), (C1,F8)}, based on different random
choices for $C$.

For R300, R700, R1000, R1400, R3000, and R5500, where we have more detailed information
about the roles played by the routers, we constructed instances, one for each topology,
that took those roles into account: routers most closely associated with network
customers (access routers, etc.) served as the customer vertices,
and the potential facility locations consisted of these, together with all other routers
except for the backbone routers (thus modeling the likely situation in which extra monitoring
equipment cannot easily be connected to the backbone routers).
The number of customer vertices ranged between 50\% and 93\% of the total,
and the number of vertices that were neither customers nor potential facility locations
ranged from 2\% to 41\%.

\vspace{-.1in}
\subsection{Instance Properties}\label{propertysect}

\begin{table}
\begin{center}
{\Large Synthetic Instances}
\bigskip
{
\begin{tabular}{c|r @{\ \ \ \ \ \ \ \ } r @{\ \ \ \ \ } r @{\ \ \ \ \ } r @{\ \ \ \ \ } r @{\ \ \ \ \ } r @{\ \ \ } r @{\ \ } r}
Class & 50 & 100 & 190 & 220 & 250 & 300 & 558 & 984\\ \hline
\multicolumn{9}{c}{\rule[-.2cm]{0cm}{.6cm}Average Number of Triples in Thousands}\\ \hline
C1,F1  & \ 12.9 & 57.8 & \ \ 224.8 & \ \ \ \ \ 340 & \ \ \ \ \ 487 & \ \ \ \ \ 756 & \ \ \ \ \ 4,111 & 19,954 \\
C2,F1  &  6.3 & 31.6 & 119.1 & 184 & 270 & 375 & 2,070 & 10,310 \\
C4,F1  &  3.5 & 15.1 &  53.6 &  92 & 118 & 185 & 1,131 &  5,136 \\
C8,F1  &  2.2 &  7.5 &  28.0 &  44 &  59 &  98 &   471 &  2,642 \\
C2,F2  &  1.5 &  6.7 &  26.9 &  43 &  55 &  93 &   507 &  2,438 \\
C4,F4  &  0.2 &  0.8 &   3.7 &   5 &   8 &  11 &    62 &    299 \\
C8,F8  &  0.0 &  0.1 &   0.4 &   1 &   1 &   1 &     6 &     38 \\ \hline
\multicolumn{9}{c}{\rule[-.2cm]{0cm}{.6cm}Average Percent of Total Number of Possible Triples for Set-Disjoint Instances}\\ \hline
C1,F1  & 21.99 & 11.92 &  6.66 &  6.48 &  6.30 &  5.65 &  4.76 & 4.20\\
C2,F1  & 21.58 & 13.02 &  7.05 &  7.00 &  6.98 &  5.61 &  4.79 & 4.34\\
C4,F1  & 22.63 & 12.43 &  6.28 &  7.00 &  6.09 &  5.53 &  5.22 & 4.33\\
C8,F1  & 26.42 & 11.92 &  6.56 &  6.60 &  5.94 &  5.80 &  4.35 & 4.45\\
C2,F2  & 22.03 & 11.33 &  6.48 &  6.59 &  5.73 &  5.65 &  4.72 & 4.12\\
C4,F4  & 22.03 & 11.30 &  7.20 &  6.50 &  6.82 &  5.63 &  4.64 & 4.07\\
C8,F8  & 20.10 & 15.23 &  7.32 &  7.24 &  6.99 &  5.57 &  3.83 & 4.16\\ \hline
\multicolumn{9}{c}{\rule[-.2cm]{0cm}{.6cm}Average Percent Increase in Triples for Pathwise-Disjoint Instances over Set-Disjoint Ones}\\ \hline
C1,F1  &  6.94 &  6.63 &  4.96 &  5.99 &  3.88 &  6.68 &  4.29 & 5.60\\
C2,F1  &  6.93 &  5.73 &  5.05 &  5.10 &  3.25 &  5.99 &  4.09 & 4.90\\
C4,F1  &  7.52 &  5.00 &  4.87 &  6.08 &  3.84 &  7.10 &  2.74 & 4.23\\
C8,F1  &  5.28 &  5.14 &  3.93 &  7.49 &  3.26 &  5.60 &  3.37 & 2.88\\
C2,F2  &  8.32 &  7.13 &  4.36 &  5.37 &  4.50 &  6.62 &  3.88 & 6.15\\
C4,F4  &  3.99 &  6.98 &  7.46 &  8.12 &  4.38 &  7.60 &  4.97 & 5.69\\
C8,F8  & 17.71 &  7.20 &  5.61 &  3.32 &  1.98 &  5.63 &  4.62 & 6.15\\ \hline
\multicolumn{9}{c}{\rule[-.2cm]{0cm}{.6cm}Average Percent Decrease in Triples for Equal-Weights on Set-Disjoint Instances} \\ \hline
C1,F1 & 16.88 & 13.82 &  9.07 & 10.29 & 14.53 & 13.22 & 26.89 & 31.13\\
C2,F1 & 15.48 & 15.12 & 10.48 &  9.67 & 15.11 & 14.89 & 26.68 & 31.09\\
C4,F1 & 11.82 & 15.48 &  9.52 &  9.79 & 14.50 & 14.44 & 30.95 & 32.24\\
C8,F1 & 22.68 &  8.64 &  4.92 &  7.68 & 12.71 & 15.31 & 24.03 & 39.41\\
C2,F2 & 13.81 & 14.27 &  4.47 & 13.34 & 11.84 & 11.26 & 26.48 & 30.23\\
C4,F4 & 13.60 & 13.03 &  7.46 &  9.13 & 20.49 &  3.16 & 24.81 & 32.15\\
C8,F8 &  8.35 & 17.06 &  5.55 & 19.95 & 14.86 &  4.89 & 27.72 & 37.81\\ \hline
\end{tabular}
}
\caption{Average numbers of triples for synthetic instances, the percentage of
all possible triples that these numbers represent, the percentage increases
in numbers of triples for the arc-disjoint version of
pathwise-disjoint instances, and for equal-weight instances.}\label{tripletab}
\end{center}
\vspace{-.25in}
\end{table}

\vspace{-.1in}

\begin{table}
\begin{center}
{\Large ISP-Based Instances}
\bigskip
{
\begin{tabular}{c|r @{\ \ \ \ } r @{\ \ \ } r @{\ \ \ } r @{\ \ \ } r @{\ \ \ } r @{\ \ \ \ \ } r @{\ \ \ \ } r @{\ \ \ \ } r @{\ \ \ \ } r}
Class & R100a & R100b & R200 & R300 & R500 & R700 & R1000 & R1400 & R3000 & R5500\\ \hline
\multicolumn{11}{c}{\rule[-.2cm]{0cm}{.6cm}Average Number of Triples in Millions}\\ \hline
C1,F1  & 0.22 & 0.36 &  1.44 &  2.47 & 17.7 &  22.0 & 47.2 & 134 & 816 & 8,156\\
C2,F1  & 0.11 & 0.18 &  0.73 &       &  8.8 &       &      &     &     &     \\
C4,F1  & 0.06 & 0.09 &  0.36 &       &  4.4 &       &      &     &     &     \\
C8,F1  & 0.03 & 0.05 &  0.20 &       &  2.2 &       &      &     &     &     \\ \hline
\multicolumn{11}{c}{\rule[-.2cm]{0cm}{.6cm}Average Percent of Total Number of Possible Triples for Set-Disjoint Instances}\\ \hline
C1,F1  & 30.0 & 34.0 &  25.6 &  33.7 &  23.9 &  28.2 &  11.2 & 21.4 & 35.0 & 20.9\\
C2,F1  & 30.1 & 34.1 &  26.0 &       &  23.5 &       &       &      &      &     \\
C4,F1  & 32.2 & 32.7 &  25.4 &       &  23.6 &       &       &      &      &     \\
C8,F1  & 26.5 & 32.8 &  27.0 &       &  23.3 &       &       &      &      &     \\ \hline
\multicolumn{11}{c}{\rule[-.2cm]{0cm}{.6cm}Average Percent Increase in Triples for Pathwise-Disjoint Instances over Set-Disjoint Ones}\\ \hline
C1,F1  & 43.6 & 40.1 &  59.8 &  41.7 &  24.9 &  31.0 &  335.6 & 33.0 & 25.9 & 101.7\\
C2,F1  & 42.6 & 39.9 &  58.4 &       &  25.3 &       &       &      &      &     \\
C4,F1  & 41.2 & 46.5 &  63.6 &       &  26.7 &       &       &      &      &     \\
C8,F1  & 53.4 & 36.2 &  56.8 &       &  24.0 &       &       &      &      &     \\ \hline
\multicolumn{11}{c}{\rule[-.2cm]{0cm}{.6cm}Average Percent Increase in Triples for Arc-Disjoint over Vertex Disjoint Instances} \\ \hline
C1,F1  & 1.0 & 1.3 &  1.3 &  .1 &  1.8 &  .3 &  15.5 & 1.5 & .1 & 5.5\\
C2,F1  & 1.0 & 1.3 &  1.2 &       &  1.9 &       &       &      &      &     \\
C4,F1  & 1.4 & 1.4 &  1.3 &       &  1.9 &       &       &      &      &     \\
C8,F1  & .8 & 1.6 &  1.3 &       &  .9 &       &       &      &      &     \\ \hline
\multicolumn{11}{c}{\rule[-.2cm]{0cm}{.6cm}Average Percent of Total Number of Possible Triples for Arc-Path-Disjoint Instances} \\ \hline
C1,F1  & 43.5 & 48.3 &  41.5 & 47.8 &  30.3  &  37.0 &  56.4 &  28.9 & 44.1 & 44.3\\
C2,F1  & 43.2 & 48.3 &  41.6 &       &  30.0 &       &       &      &      &     \\
C4,F1  & 46.2 & 48.4 &  42.0 &       &  30.4 &       &       &      &      &     \\
C8,F1  & 40.9 & 45.1 &  42.4 &       &  29.2 &       &       &      &      &     \\ \hline
\multicolumn{11}{c}{\rule[-.2cm]{0cm}{.6cm}Average Percent Decrease in Triples for Equal-Weights on Set-Disjoint Instances} \\ \hline
C1,F1  & 27.4 & 16.0 &  4.9 & 29.4 &  14.2  &  31.6 &  -16.7 &  15.7 & 55.1 & 58.7\\
C2,F1  & 26.1 & 16.1 &  5.7 &       &  13.7 &       &       &      &      &     \\
C4,F1  & 27.0 & 11.9 &  6.3 &       &  12.0 &       &       &      &      &     \\
C8,F1  & 30.1 & 15.3 &  2.7 &       &  16.9 &       &       &      &      &     \\ \hline
\end{tabular}
}
\caption{Average percentage of valid triples for real world instances,
the percentage increases in numbers of triples for the vertex- and arc-disjoint versions of
pathwise-disjoint instances, and for equal-weight instances.}\label{realtripletab}
\end{center}
\vspace{-.25in}
\end{table}

Our synthetic instances reflect the structure of real-world LANs and WANs, 
and consist of many, relatively small 2-connected components.
For instances with $|V| \geq 100$, the largest 2-connected component averages
around 12\% of the vertices
(112.5 vertices for our $|V| = 984$ instances).
In contrast, the Tier-1 ISP instances
consist of one very large 2-connected component, typically containing over 95\%
of the vertices, with small 2-connected components hanging off of it,
most of size 2.
One consequence of this difference in topology
is differing numbers of triples for the two instance types.

For a given instance with a set $C$ of customers and a set $F$ of potential
facility locations, there are $|C|(|F|-1)(|F|-2)/2$ potential triples $(c,f,f')$,
each one containing a customer $c$ and two distinct potential
facility locations in $F - \{c\}$,
For our synthetic instances, only a relatively small fraction of the potential
triples are actually valid triples in the set-disjoint case.
For each class $(Cx,Fy)$, the percentage of valid triples declines as $N$ increases,
until it appears to stabilize somewhere between 4 and 5\%.
See Table \ref{tripletab}.
The ISP-based instances yield substantially more set-disjoint triples
than the synthetic ones of the same size.
The percentage of valid triples, with one exception, ranges between 20 and 35\%,
the exception being R1000, which at 11.2\% still has
more than double the percentage for the large synthetic instances.
See Table \ref{realtripletab}. 
Note that, because of the cubic growth in the number of potential triples, the
actual number of triples can become quite large.
For our largest instance, R5500, there are over 5.7 billion valid triples in the
set-disjoint case.

The choice of arc weights also has an effect on our instances.
The optimized weights of our synthetic instances lead to relatively few
shortest-path ties, and so the pathwise-disjoint instances typically have only
5 to 6\% more valid triples than the corresponding set-disjoint
instances.
The weights in the ISP instances were constrained by a variety of factors
in addition to traffic balancing, and yield far more ties.
As a result, the number of path-disjoint triples exceeds the number of
set-disjoint triples by from 24 to 335\%, that latter occurring for our anomalous
R1000 instance.

We also measured the effect of relaxing the constraints on our path-disjoint instances,
by requiring only that the paths be arc-disjoint, not vertex-disjoint.
(Although vertex-disjointness adds robustness to our monitoring application,
the mathematics of the application only requires arc-disjointness.)
Here the effects were relatively minor for both synthetic and ISP-based instances,
typically yielding an increase of less than 2\% in the number of valid triples.
Instance R1000 was again the major exception, yielding an increase of 15\%, to reach
a total of more than half the maximum possible number of triples.

Finally, to further explore the effect of increased numbers of ties
in the set-disjoint case, we considered
modified versions of our instances in which all edge weights are set to 1.
As might be expected for such weights,
the number of set-disjoint triples drops substantially for almost all
our synthetic and ISP-based instances (the only exception again being instance
R1000, where we saw a 16\% increase).

\section{Summary of Experimental Results}\label{section:experiments}

This section evaluates our algorithms in three areas.

\begin{itemize}
\item
\textsl{Accuracy:} How close to optimal were the solutions
provided by each of heuristics \greedy, \shs,  and \dhs?
For small instances, for which the exact optimal can be computed using CPLEX,
we can compare to that.
In the set-disjoint case, where the SHS lower bound turns out to equal the
optimal solution value for all our instances, we can compare to that.
For large path-disjoint instances, we have no good, readily computable,
standard of comparison.
We consider two options.
First, we can compare the results to the (very weak but
linear-time computable) lower bound obtained when one requires only that paths
be vertex-disjoint, not that they be shortest paths.
Second,
we can evaluate the results for {\gr4} in the path-disjoint case by
attempting to improve on them
using the {\genetic} algorithm and large amounts of computing time.
(Note that {\genetic} optimally solved all the bad set-disjoint instances
for {\gr4} except R5500, which is too big for {\genetic} to handle.)

\item
\textsl{Running Time:} How fast are the algorithms?
Our runs were primarily performed on two
machines.
We used a shared-memory machine with 1.5 Ghz Itanium processors and
256 Gigabytes of random access memory, running
SUSE LINUX with kernel 2.6.16.60-0.81.2-default,
for the heuristic runs on instance R5500, where 80 gigabytes of main
memory were needed because of the large number of triples involved.
In the text this will be referred to as ``our slower machine.''
The remainder of our runs were primarily performed on a shared
memory multiprocessor machine with 24 2.7 Ghz Intel Xeon processors
and 50 gigabytes of random access memory running
CentOS Linux with kernel 2.6.32-431.11.2.el6.x86\_64Linux.
In addition, we did some confirmatory runs on a MacBook Pro with a
2.4 Ghz dual-core Intel i5 processor and 8 gigabytes of RAM, running
Mac OS 10.9.4 and with Gurobi 5.6 replacing CPLEX 11.0 as our MIP solver.
When we report running times, we give the sum of the system and user times,
and will specify the machine used.
(Each algorithm run used only a single processor/core of the
respective machine.)
All our codes were written in {\tt C}, and compiled using {\tt gcc}.
Unless otherwise stated, the running times reported are for code compiled
with no optimization flags set, although, where
appropriate, we will indicate the kinds of
speedups one might expect using {\tt -O2} optimization.

It should be noted that on none of these machines were running times
precisely repeatable.
Differences of 20\% or more on duplicate runs were not uncommon.
So our running time reports can provide only order-of-magnitude information,
and can provide evidence that one code is faster than another only if the
differences are substantial and hold across many instances.

\item
\textsl{Solution Quality (Cost Reduction):}
The first two measures address the quality of the algorithms.
This one addresses what they tell us about the applications.
Even an optimal solution will not be enough to justify the proposed
monitoring scheme of \citet{GBDS-loss08} if
it does not provide a significant savings over the default solution
that simply takes {\em all} customers (which for our instances are
themselves facility locations) and is likely to yield more reliable measurements.
We also consider how much further improvement is obtainable when going
from the setwise-disjoint to the pathwise-disjoint versions of the problem.
\end{itemize}

\noindent

\subsection{Accuracy and Running Times for the Set-Disjoint Case}
\subsubsection{Optimization via our {\mip} Code}
In Table \ref{lb+opttab}, we summarize the results for our {\mip}
optimization code on our synthetic and ISP-based set-disjoint instances.
All the runs were performed on the faster of our two Linux machines.
The first section of the table reports how many of the 10 instances of
each type that the code was able to solve within a 24-hour time bound.
For $|V|=190$, we eventually solved the missing (C4,F1)
and (C8,F1) instances, although the longest (C4,F1) took
slightly more than a week.
We also solved the two missing (C4,F1) instances based on the ISP-based
R200 instances, although the longer took over a week.
On the other hand, for 299 of the 383 solved instances, {\mip}
took less than 10 minutes.
The second section of the table reports the worst running times
(system time plus user time, rounded to the nearest second)
for the classes where all instances were solved in less than 24 hours.
For our testbed in general,
the code tends to run into trouble as the number of
customers grows beyond 100, with exceptions when the number of
facilities is also relatively small.
Its running times seem clearly to be growing in exponential fashion as $|C|$
increases.
Note also that,
when all vertices are potential facility locations, reducing the
number that are also customers can actually make the instance harder.

\begin{table}[t]
\begin{center}

\bigskip
{
\begin{tabular}{c|r @{\ \ } r @{\ \ \ } r @{\ \ \ } r @{\ \ \ } r @{\ \ \ } r @{\ \ \ } r @{\ \ \ } r|  r  r  r }
\multicolumn{9}{c}{\rule[-.2cm]{0cm}{.6cm}\ \ \ \ \ \ \ \ Synthetic Instances} & \multicolumn{3}{c}{\rule[-.2cm]{0cm}{.6cm}ISP-Based Instances}\\ \hline
Class & 50 & 100 & 190 & 220 & 250 & 300 & 558 & 984 & R100a & R100b & R200\\ \hline
\multicolumn{12}{c}{\rule[-.2cm]{0cm}{.6cm}Number of Instances Solved in Less Than 24 Hours by MIP}\\ \hline
(C1,F1) & 10 &  10 &   10 &  - &  - &  - &  - &  - & 1 & 1 & 1\\
(C2,F1) & 10 &  10 &   10 &  - &  - &  - &  - &  - & 5 & 5 & -\\
(C4,F1) & 10 &  10 &    9 &  - &  - &  - &  - &  - & 5 & 5 & 3\\
(C8,F1) & 10 &  10 &    9 &  - &  - &  - &  - &  - & 5 & 5 & 5\\ \cline{10-12}
(C2,F2) & 10 &  10 &   10 & 10 & 10 & 10 &  - &  -  \\
(C4,F4) & 10 &  10 &   10 & 10 & 10 & 10 & 10 & 10 \\
(C8,F8) & 10 &  10 &   10 & 10 & 10 & 10 & 10 & 10 \\ \hline
\multicolumn{12}{c}{\rule[-.2cm]{0cm}{.6cm}Worst Running Times (in Seconds) for Instances of Each Completed Class} \\ \hline
(C1,F1) & 20 &  89 &   23 & 518 & 42,998 &  - &  - &  - & 0 & 0 & 74\\
(C2,F1) & 18 &  86 & 1725 &   - &  - &  - &  - &  - &   8 &  1 & -\\
(C4,F1) & 36 &  46 &    - &   - &  - &  - &  - &  - & 291 & 92 & -\\
(C8,F1) & 35 & 803 &    - &   - &  - &  - &  - &  - &  13 & 46 & 33,789\\ \cline{10-12}
(C2,F2) &  0 &   0 &    0 &  0 &  2 & 25 &  - &   - \\
(C4,F4) &  0 &   0 &    0 &  0 &  0 &  0 &  0 & 227 \\
(C8,F8) &  0 &   0 &    0 &  0 &  0 &  0 &  0 &   0 \\ \hline
\end{tabular}
}
\caption{Results for our {\mip} optimization code on set-disjoint instances.}\label{lb+opttab}
\end{center}
\vspace{-.25in}
\end{table}

\subsubsection{Lower Bounds and Solutions via our {\lb} Code}
Fortunately, the exponential behavior of our {\mip} code is not an
insurmountable difficulty in the case of our set-disjoint instances.
Not only did our {\lb} code produce a lower bound that equals the optimal
solution value for all our instances, but the hitting set it constructed
was always itself a feasible, and hence optimal, solution.
Moreover, it took no more than 4 seconds on any instance besides R5500,
for which it only took 22 seconds.  In general, it was far faster than
any of our heuristics, often by a factor of 100 or more,
even though it involves finding an optimal solution
to an instance of the NP-hard {\sc Hitting Set} problem.
Moreover, its space requirements are much less than those of our heuristics,
given that the (Customer,Neighbor,Facility) triples it stores are
far fewer in number than the standard (Customer,Facility,Facility) triples
required by the heuristics (and stored by them in triplicate).
Consequently, for our largest instance, R5500, it required just over 1 GB,
whereas the heuristics need in excess of 65 GB.

\subsubsection{The Heuristics}
As for the heuristics, the ``best-of-400-iterations'' versions
of {\greedy}, {\shs}, and our two variants on {\dhs}
also all performed well in general on our test bed, although two
instances were seriously problematical for at least one of them.

The best solutions found (over 400 runs) by {\shs}
were optimal for {\em all} the instances in our testbeds, both synthetic and
ISP-based.
{\dhs}$_H$ and {\greedy} performed almost as well.
The former failed to produce an optimal solution for just two instances,
on one missing optimal by just a little, but on the other missing by a lot.
Its solution for one of our (C1,F1), $|V|=984$ synthetic instances was
one facility larger than optimal, but its solution for the ISP-based
instance R1000 was off by more than a factor of 3.
{\greedy} was optimal for all our synthetic instances and all but
three of our ISP-based instances.
It was slightly worse than {\dhs}$_H$ on R1000 and off by a factor of 2.2
on R5500.
Its third nonoptimality was minor, consisting of a solution
with one more facility than optimal on the second of five (C2,F1)
variants on R500.
Although {\dhs}$_1$ found optimal solutions on all our ISP-based instances,
if failed to find optimal solutions for seventeen of our synthetic instances,
six of the ones with $|V|=558$ and eleven of those with $|V|=984$,
including the one that on which {\dhs}$_H$ failed.
On twelve it was off by one facility, on four it was off by two,
and on one it was off by five.
In no case, however, was the discrepancy greater than 3.3\%.
The better of the {\dhs}$_H$ and {\dhs}$_1$ solutions was always optimal.

To further compare these heuristics, let us consider two other metrics,
running time and robustness, together with the tradeoffs between the two.
See Table \ref{heurtimetab}, which reports the running times in seconds for
the 400-iteration versions of our heuristics on our faster Linux machine
(80-iteration runs on our slower machine for R5500, which would not fit
on the faster machine), and on our MacBook Pro, where Gurobi 5.6 replaced
CPLEX 11.0 as our MIP solver, and we compiled with the -O2 optimization flag
set.
The instances covered are the largest ones in our testbed for each
ISP topology with $|V| \geq 200$, i.e., the single instance for each
topology where we knew the OSPF weights and router types,
and the (C1,F1) instance for each of the instances where we optimized the
weights ourselves.
For comparison purposes, the table also breaks out the time devoted
to constructing the triples used by the heuristics, and the times for
our code that computes the {\lb} and determines whether the optimal hitting
set produced is a feasible solution.
Running times for our synthetic instances of similar sizes were roughly
comparable.

\begin{table}
\begin{center}
{
\begin{tabular}{r|r @{\ \ \ \ } r @{\ \ \ } r @{\ \ \ } r @{\ \ \ } r @{\ \ \ } r @{\ \ \ \ \ } r @{\ \ \ \ } r r}
Algorithm  & R200 & R300 & R500  & R700 & R1000 & R1400 & R3000 & R5500*\\ \hline
\multicolumn{9}{c}{\rule[-.2cm]{0cm}{.6cm}Linux Machine(s), Compiled without Optimization} \\ \hline
\lb         &  0.1 &  0.1 &   0.4 &  0.3 &   1 &   1 &    4 &     23\\
Triple Gen  &  0.1 &  0.2 &   0.9 &  1.4 &   4 &   8 &   41 &  1,800\\
\shs        &  7.3 &  4.7 &  68.3 & 32.0 & 125 & 361 &  452 & 23,170\\
{\dhs}$_1$  &  7.2 &  6.7 &  73.1 & 37.9 & 147 & 389 &  524 & 24,560\\
{\dhs}$_H$  &  8.4 & 11.6 &  74.1 & 75.8 & 396 & 730 & 1165 & 47,970\\
\greedy     & 11.8 & 14.0 & 116.1 & 90.9 & 578 & 704 & 2171 & 70,100\\ \hline
\multicolumn{9}{c}{\rule[-.2cm]{0cm}{.6cm}MacBook Pro, Compiled with Optimization} \\ \hline
\lb         &  0.8 &  0.7 &   1.2 &  1.0 &   1 &   1 &    3 & 14\\
Triple Gen  &  0.5 &  0.5 &   0.8 &  0.9 &   2 &   3 &   14 &  -\\
\shs        &  2.2 &  2.0 &  17.7 &  8.8 &  32 &  94 &  179 &  -\\
{\dhs}$_1$  &  2.4 &  2.5 &  19.0 & 13.4 &  44 & 113 &  265 &  -\\
{\dhs}$_H$  &  2.5 &  2.3 &  19.8 & 14.9 & 130 & 119 &  277 &  -\\
\greedy     &  3.2 &  4.1 &  30.3 & 25.9 & 188 & 198 &  866 &  -\\ \hline

\end{tabular}
}
\caption{Running times in seconds for triple construction, full
400-iteration heuristic runs, and {\lb} computations.  *For R5500, all Linux times except those for {\lb} are on the slower
of our two machines, because of memory constraints, and for 80-iteration runs.
Memory constraints also prevented running anything except {\lb} on the MacBook Pro for this instance. }\label{heurtimetab}
\end{center}
\vspace{-.25in}
\end{table}

The algorithms are typically ranked {\shs}, {\dhs}$_1$ {\dhs}$_H$,
{\greedy} in order of increased running times for both machines,
although there are minor inversions for a few instances.
The large differences typically seen between {\dhs}$_H$ and {\dhs}$_1$ on the
Linux machines occurred for only one instance (R1000) on the MacBook Pro runs.
The fact that, for the smaller instances, the
running times for {\lb} and triple generation are larger on the
optimized MacBook Pro runs than on the unoptimized Linux runs
is attributable to larger system times for triple construction.
For the (facility,neighbor,facility) triples used in computing the {\lb},
the system time was roughly 0.7 seconds, independent of instance size,
on the MacBook Pro, versus 0.1 seconds on the Linux machine, and the
differences in system times for computing standard triples
in our heuristics were similar.
It is interesting to note that, on both machines, the growth in user time
for triple construction appears to be somewhat slower than our
theoretical bound of $\Theta(|V|^3)$, although
no firm conclusions can be drawn given the differences in topology and
weight settings for the various instances.

The general take-away from these results, however, is that all of our
algorithms have running times that would be practical in the context of
our proposed applications, where the time for implementing a solution
will still greatly exceed the time to compute it.
The main obstacle to practicality is the memory requirement for
instances significantly larger than R3000, and in the final
section of this paper, we will discuss ways that this memory requirement
might be reduced while keeping running times practical.

Of course, although running time is not a substantial obstacle, it is always
worthwhile to find ways to reduce it
(beyond simply using optimized compilation).
One obvious way would be to perform fewer than 400 iterations.
Conversely, we might want give up running time and perform more iterations,
if that would improve the robustness of our results.
How robust are they?

\begin{table}
\begin{center}
{
\begin{tabular}{r|c c c c c}
k  & $N=100$ & $N=200$ & $N=400$ & $N=800$ & N=1600\\ \hline
 1 &  0.77855704 &  0.60615106 &  0.36741911 &  0.13499680 &  0.01822414\\
 2 &  0.60577044 &  0.36695782 &  0.13465804 &  0.01813279 &  0.00032880\\
 3 &  0.47103323 &  0.22187230 &  0.04922732 &  0.00242333 &  0.00000587\\
 4 &  0.36603234 &  0.13397967 &  0.01795055 &  0.00032222 &  0.00000010\\
 5 &  0.28425652 &  0.08080177 &  0.00652893 &  0.00004263 &  0.00000000\\
 6 &  0.22060891 &  0.04866829 &  0.00236860 &  0.00000561 &  0.00000000\\
 7 &  0.17110232 &  0.02927600 &  0.00085708 &  0.00000073 &  0.00000000\\
 8 &  0.13261956 &  0.01758795 &  0.00030934 &  0.00000010 &  0.00000000\\
 9 &  0.10272511 &  0.01055245 &  0.00011135 &  0.00000001 &  0.00000000\\
10 &  0.07951729 &  0.00632300 &  0.00003998 &  0.00000000 &  0.00000000\\
11 &  0.06151216 &  0.00378375 &  0.00001432 &  0.00000000 &  0.00000000\\
12 &  0.04755251 &  0.00226124 &  0.00000511 &  0.00000000 &  0.00000000\\
13 &  0.03673647 &  0.00134957 &  0.00000182 &  0.00000000 &  0.00000000\\
14 &  0.02836164 &  0.00080438 &  0.00000065 &  0.00000000 &  0.00000000\\
15 &  0.02188134 &  0.00047879 &  0.00000023 &  0.00000000 &  0.00000000\\ \hline

\end{tabular}
}
\caption{Probabilities of failure to find an optimal solution
in $N$ independent runs when the probability of finding one in 400 runs is
$k/400$.}\label{probtab}
\end{center}
\vspace{-.25in}
\end{table}

In the context of our experiments, let $k$ be the number of runs
(out of our 400) that returned an optimal solution.
We can then roughly estimate
the probability of finding an optimal solution to be $k/400$,
and the probability that such a solution will {\em not} be found if
we perform another $N$ independent random runs is $((400-k)/400)^N$.
Table \ref{probtab} gives the probabilities for $1 \leq k \leq 15$ and
$N \in \{100,200,400,800,1600\}$.  Using this table, we can evaluate the
likely robustness of our heuristics, based on the instances with the 
smallest values of $k$ for each heuristic.
If we want to reduce the failure probability to, say, 0.1\%, we would
need to have observed a 400-iteration success count of 2 if we were going
to perform 1600 iterations, 4 if 800 iterations, 7 if 400 iterations,
14 if 200 iterations, and more than 15 if 100 iterations (actually, a
count of 27 would suffice in this case).
Given that the counts given by single 400-iteration runs give noisy
estimate of the true probabilities, it might be safer to inflate the
above count thresholds in practice.

To help us evaluate our four heuristics in this context,
Table \ref{heurprobtab}
gives the 15 smallest success counts for each, both for synthetic
and ISP-based instances, together with derived counts for an algorithm
{\dhs}$_{1H}$, which performs 200 runs each for {\dhs}$_1$ and {\dhs}$_H$,
and reports the best solution overall.
Note that the entries for the derived heuristic are not simply the averages of the values for {\dhs}$_1$ and {\dhs}$_H$ in the same row of the table,
since these may represent counts for different instances -- the hardest
instances for one of the two need not be the same as those for the other.
The counts tend to be higher for our ISP-based instances, since 15 represents
over 21\% of our 70 instances of that type, whereas it represents only
about 2.4\% of our 630 synthetic instances,

Based on the results in the table, one can conclude that {\greedy} is the
most robust for synthetic instances, with {\shs} almost as good.
Neither {\dhs}$_1$ nor {\dhs}$_H$ is competitive for such instances, and
their combination seems less robust than {\dhs}$_H$, the better of the two,
by itself.
In contrast, for the ISP-based instances, the combination is distinctly
more robust than either {\dhs}$_1$ or {\dhs}$_H$ alone, and also dominates
{\greedy}, although {\shs} is the best overall, having only one count
less than 25.
The corresponding instance, where {\shs} found an optimal solution only once in
400 times, is troublesome, however.
This suggests performing a hybrid of \shs\ with one of our other three heuristics
(again running each 200 times and taking the best).
If we blend with {\greedy}, the smallest counts are 61 (synthetic) and
6.0 (ISP-based).
Blending with {\dhs}$_1$ yields corresponding smallest counts of 34.5 and 9.0,
and blending with {\dhs}$_H$ yields smallest counts of 46.5 and 10.0.
The last choice would seem to be the best overall, as
400 total iterations should likely be enough to keep the
failure rate for all instances less than 0.1\%, and, depending on the
accuracy of our sample, 200 iterations might suffice.
All three blends, however, offer significant improvements in robustness over
any of the individual heuristics.

\begin{table}
\begin{center}
{
\begin{tabular}{r|r r r r r | r r r r r}
 & & \multicolumn{3}{c}{Synthetic Instances} & & & \multicolumn{3}{c}{ISP-Based Instances} & \\
k & {\shs} & {\dhs}$_1$ & {\dhs}$_H$ & {\dhs}$_{1H}$ & {\greedy} & {\shs} & {\dhs}$_1$ & {\dhs}$_H$ & {\dhs}$_{1H}$ & {\greedy}\\ \hline
 1 & 10 & 0 &  0 & 0.0 & 10 &   1 &   3 &   0 &   6.0 &  0 \\
 2 & 10 & 0 &  1 & 0.5 & 19 &  25 &   5 &   0 &   8.5 &  0 \\
 3 & 13 & 0 &  1 & 1.0 & 29 &  27 &  13 &   0 &  20.5 &  7 \\
 4 & 16 & 0 &  4 & 2.0 & 50 &  35 &  23 &   1 &  21.0 &  7 \\
 5 & 39 & 0 &  4 & 2.0 & 52 &  68 &  37 &   7 &  26.5 & 29 \\
 6 & 47 & 0 &  5 & 2.5 & 55 &  73 &  40 &  14 &  40.0 & 29 \\
 7 & 52 & 0 &  5 & 3.5 & 61 &  80 &  44 &  28 &  41.5 & 33 \\
 8 & 57 & 0 &  6 & 3.5 & 65 &  91 &  49 &  30 &  44.0 & 35 \\
 9 & 57 & 0 &  8 & 4.0 & 66 &  95 &  60 &  44 &  58.0 & 53 \\
10 & 60 & 0 &  8 & 4.0 & 69 & 101 &  67 &  46 &  68.0 & 57 \\
11 & 64 & 0 &  9 & 4.5 & 72 & 110 &  68 &  56 &  79.5 & 66 \\
12 & 65 & 0 &  9 & 5.0 & 74 & 136 &  80 &  68 &  87.0 & 85 \\
13 & 65 & 0 & 10 & 5.0 & 79 & 139 & 100 &  93 & 100.0 & 88 \\
14 & 67 & 0 & 10 & 5.5 & 80 & 163 & 107 & 100 & 100.0 & 89 \\
15 & 67 & 0 & 10 & 6.0 & 82 & 177 & 123 & 100 & 103.5 & 98 \\ \hline

\end{tabular}
}
\caption{The 15 smallest numbers of optimal solutions found in 400
independent runs over synthetic and ISP-based instances.}\label{heurprobtab}
\end{center}
\vspace{-.25in}
\end{table}

\subsubsection{Challenging {\lb}: Equal-Weight Edges}

Given the fact that our {\lb} code produced hitting sets that were optimal
solutions for all 700 instances in our set-disjoint testbed, one might
wonder why we should bother with heuristics at all for the set-disjoint case.
Perhaps instances where the {\lb} is smaller than the optimal solution value,
or where the {\lb} solution is otherwise infeasible,
simply do not arise in practice, possibly because of the
restricted topologies in real world LAN/WANs and ISP networks, or because
edge weights are optimized to manage traffic.
To investigate the latter possibility, we repeated our experiments with
all the weights set to 1 -- not a very good choice for traffic management,
but certainly a natural default, and one that yields a major change in
the set of triples in the derived ``cover-by-pairs'' instances, as observed
in Section \ref{propertysect}.

With these instances, the hitting sets constructed by our {\lb} code were still
legal and hence optimal covers for all 70 of our ISP-based instances.
However, for 51 of our 630 synthetic instances, the hitting set was
{\em not} a legal cover, and for nine of these, the best solution found
by our 400-iteration {\shs} run was still larger than the computed lower bound
(although never by more than 2.1\%).
These were all (C1,F1) instances, five with $|V|=984$, three with $|V|=558$,
and one with $|V|=250$, and in all nine cases, the {\lb} nevertheless continued
to equal the optimal solution size.
The optimal solutions for these nine were {\em not} found by any of our
400-iteration runs of {\dhs}$_1$ or {\dhs}$_H$.
As in our experiments on the synthetic instances in our standard testbed,
it was {\greedy} that proved to have the best synergies with {\shs}.
Our 400-iteration runs of {\greedy} found solutions
matching the {\lb} for three of the 984's, for one of the 558's, and for the 250.
For the remaining four instances, we needed to increase the number of iterations.
For the two remaining 984's, 4000-iteration runs of {\shs} found solutions
matching the {\lb}.
For one of the two remaining 558's, a 100,000-iteration run of {\greedy}
found a lower-bound matching solution (three times),
although a 200,000-iteration of {\shs} failed to find any.
For the other remaining 558, a 100,000-iteration run of {\shs}
found a lower-bound matching solution (twice),
although a 200,000-iteration run of {\greedy} failed to find any.

Thus the track record of {\lb} equaling the optimal solution value for
``real world'' WAN and ISP topologies remains unblemished.
We can conclude, however, that we should not count on
our {\lb} code always producing feasible optimal solutions,
and that {\shs} and {\greedy}
remain valuable, and complementary, backups.

\begin{table}[t]
\begin{center}

\bigskip
{
\begin{tabular}{c|r @{\ \ } r @{\ \ \ } r @{\ \ \ } r @{\ \ \ } r @{\ \ \ } r @{\ \ \ } r @{\ \ \ } r|  r  r  r }
\multicolumn{9}{c}{\rule[-.2cm]{0cm}{.6cm}\ \ \ \ \ \ \ \ Synthetic Instances} & \multicolumn{3}{c}{\rule[-.2cm]{0cm}{.6cm}ISP-Based Instances}\\ \hline
Class & 50 & 100 & 190 & 220 & 250 & 300 & 558 & 984 & R100a & R100b & R200\\ \hline
\multicolumn{12}{c}{\rule[-.2cm]{0cm}{.6cm}Number of Instances Solved in Less Than 24 Hours by MIP}\\ \hline
(C1,F1) & 10 & 10 & 10 & 10 & 10 &  - &  - &  - & 1 & 1 & 1\\
(C2,F1) & 10 & 10 &  9 &  - &  - &  - &  - &  - & 5 & 5 & -\\
(C4,F1) & 10 & 10 &  8 &  - &  - &  - &  - &  -  & 5 & 5 & 3\\
(C8,F1) & 10 & 10 &  9 &  - &  - &  - &  - &  -  & 5 & 5 & 5\\ \cline{10-12}
(C2,F2) & 10 & 10 & 10 & 10 & 10 & 10 &  - &  -  \\
(C4,F4) & 10 & 10 & 10 & 10 & 10 & 10 & 10 & 10 \\
(C8,F8) & 10 & 10 & 10 & 10 & 10 & 10 & 10 & 10 \\ \hline
\multicolumn{12}{c}{\rule[-.2cm]{0cm}{.6cm}Worst Running Times (in Seconds) for Instances of Each Completed Class} \\ \hline
(C1,F1) & 21 &   9 & 28 & 19,362 & 78,660 &  - &  - &  - & 0 & 0 & 146\\
(C2,F1) & 18 & 212 &  - &  - &  - &  - &  - &  - &  17 &  21 & -\\
(C4,F1) & 35 & 818 &  - &  - &  - &  - &  - &  - & 335 & 474 & -\\
(C8,F1) & 25 & 321 &  - &  - &  - &  - &  - &  - &  50 & 644 & 40,739\\ \cline{10-12}
(C2,F2) &  0 &   0 &  0 &  0 &  4 & 13 &  - &    - \\
(C4,F4) &  0 &   0 &  0 &  0 &  0 &  0 &  0 & 8106 \\
(C8,F8) &  0 &   0 &  0 &  0 &  0 &  0 &  0 &    0 \\ \hline
\end{tabular}
}
\caption{Results for our {\mip} optimization code on vertex-disjoint instances.}\label{table:pathopt}
\end{center}
\vspace{-.25in}
\end{table}

\subsection{Accuracy and Running Times for the Path-Disjoint Case}
\subsubsection{Optimization via our {\mip} Code}
Table \ref{table:pathopt} summarizes our optimization results for the
path-disjoint versions of our instances, and shares the format of
Table \ref{lb+opttab}, which covered the set-disjoint versions.
Once again, we report the number of instances of each class that were
solved to optimality within 24 hours, and, for each class with all instances
so solved, the {\em worst} running time encountered, as a measure of
robustness.
Note that we were able to solve about the same number of instances of each
type within the 24-hour time bound as we did in the set-disjoint case, with
the worst time for each type usually slightly worse.
In more extended runs, we were able to compute the optimum for the
one missing $|V|=190$, (C2,F1) synthetic instance (in 39 hours),
and the two missing R200, (C4,F1) ISP-based instances (in 49 and 77 hours).

The results for the arc-disjoint versions of our path-disjoint instances were
roughly the same as reported in Table \ref{table:pathopt} for
the vertex-disjoint versions.  We were able to solve essentially the same
types of instances as before, with times that were roughly equivalent.
The maximum for each class tended to be shorter for the arc-disjoint
versions, although not so often as to provide a statistically meaningful
distinction. 

\subsubsection{Heuristic Results}

In evaluating heuristics for the path-disjoint case, we suffer from
several limitations in comparison to the set-disjoint case.
First, the {\lb} no longer applies, so we do not have a good standard
of comparison for those instances that our {\mip} approach failed to
solve.
Second, we cannot use {\shs} or either of our {\dhs} heuristics,
which only apply to the set-disjoint case, and so only have {\gr4} and
its higher-iteration variants (including \genetic) to consider.

Fortunately, {\gr4} seems to do very well.
It finds the optimal solution on all the instances for which that is
known, and we have not been able to find any better solutions than
the ones it provides on any of the instances for which the optimum
is {\em not} known.
In particular, we performed 1000-generation runs of our {\genetic}
algorithm, on each of these instances, and only matched the
{\gr4} solutions on each.

We had hopes that, if anything could improve on {\gr4}, then {\genetic}
would, based on results we obtained in the set-disjoint case.
There, {\genetic} found optimal solutions for two of the
three ISP-based instances where {\gr4} failed to find the optimum
(the third was R5500, which was far too large for {\genetic} to handle).
Most impressively, for R1000 it found the optimum, whereas the
best of even 100,000 independent runs of {\greedy} was over three times optimal.
The failure of {\genetic} to find any improvements in the path-disjoint
case thus at least suggests that our {\gr4} path-disjoint results might
indeed be optimal across the board.

The results were also robust.
For each of our 720 synthetic instances, the best solution found was
found at least 20 times out of 400.
For each of our 69 ISP-based instances other than R5500, the best solution
found was found at least four times out of 400 (and for all but one, the count was
34 or greater).
The best solution for R5500 was found 72 times out of 80.

For the arc-disjoint versions of our path-disjoint instances, the results
were again similar, both as to accuracy and robustness.
For every instance where the optimum was known, {\gr4} found it.
For all 720 synthetic instances, the best solution found was found at least
34 times out of 400.
For all 69 ISP-bases instances other than R5500,
the best solution found was found at least
six times out of 400 (and for all but one, the count was 41 or greater).
The instance with the low count was the same for both vertex- and arc-disjoint
versions of the problem.
The best solution for R5500 was found 79 times out of 80.

\begin{table}
\begin{center}
{
\begin{tabular}{r|r @{\ \ \ \ } r @{\ \ \ } r @{\ \ \ } r @{\ \ \ } r @{\ \ \ } r @{\ \ \ \ \ } r @{\ \ \ \ } r r}
Algorithm  & R200 & R300 & R500  & R700 & R1000 & R1400 & R3000 & R5500*\\ \hline
\multicolumn{9}{c}{\rule[-.2cm]{0cm}{.6cm}Linux Machine(s), Compiled without Optimization} \\ \hline
Path-disjoint Triple Gen &  1.0 & 1.2 & 6.9 & 8.2 &  56 &  72 &  558 &  17,400\\
Arc-disjoint Triple Gen  &  0.7 & 0.6 & 4.9 & 5.4 &  33 &  39 &  262 &  11,200\\
Set-disjoint Triple Gen  &  0.1 & 0.2 & 0.9 & 1.4 &   4 &   8 &   41 &  1,800\\ \hline
Path-disjoint \gr4   & 19.7 & 22.6 & 154.8 & 130.3 & 715 & 664 & 3612 & 79,000\\
Arc-disjoint \gr4    & 20.3 & 20.3 & 154.3 &  96.1 & 767 & 548 & 2801 & 68,500\\
Set-disjoint \gr4    & 11.8 & 14.0 & 116.1 &  90.9 & 578 & 704 & 2171 & 70,100\\ \hline
\multicolumn{9}{c}{\rule[-.2cm]{0cm}{.6cm}MacBook Pro, Compiled with Optimization} \\ \hline
Path-disjoint Triple Gen &  0.7 & 0.7 & 2.7 & 3.5 & 17 & 24 & 347  &  -\\
Arc-disjoint Triple Gen  &  0.6 & 0.6 & 1.5 & 1.8 &  9 & 10 &  95 &  -\\
Set-disjoint Triple Gen  &  0.5 & 0.5 & 0.8 & 0.9 &  2 &  3 &  14 &  -\\ \hline
Path-disjoint \gr4 &  4.6 &  5.4 & 38.3 & 35.8 & 239 & 224 & 1865 & -\\
Arc-disjoint \gr4  &  4.4 &  5.2 & 36.3 & 32.7 & 254 & 198 & 1360 & -\\
Set-disjoint \gr4  &  3.2 &  4.1 & 30.3 & 25.9 & 188 & 198 &  866 & -\\ \hline

\end{tabular}
}
\caption{Running times in seconds for triple construction and {\gr4} runs on
path-disjoint instances (both vertex- and arc-disjoint versions), with results
for set-disjoint instances included for comparison.
*For R5500, all Linux times except those for {\lb} are on the slower
of our two machines, because of memory constraints, and for 80-iteration runs.
}\label{heurtimetab2}
\end{center}
\vspace{-.25in}
\end{table}

As to running times, see Table \ref{heurtimetab2}, which mimics Table
\ref{heurtimetab} in covering the maximum-customer versions
of our ISP-based instances with more than 200 vertices, and both the
times for triple-generation and the overall {\gr4} times.
As before, we present running times for unoptimized code on our Linux
machines and optimized code on our MacBook Pro.
We list times for both the vertex- and arc-disjoint versions of our
path-disjoint instances, and, for comparison purposes, copy from
Table \ref{heurtimetab} the
corresponding times for our set-disjoint instances.

Recalling from Table \ref{realtripletab} that the number of triples
associated with an instances increases as one goes from set-disjoint to
vertex-disjoint to arc-disjoint instances, one might expect running
times to behave similarly, but this is not the case.
For triple generation, the arc-disjoint triples take less time to generate
than the vertex-disjoint ones.
The explanation for relative triple-generation running times
is algorithmic.
As described in Section \ref{section:triples}, the path-disjoint
constructions should take longer than the set-disjoint ones because they
need an extra network-flow step, and the vertex-disjoint constructions
take longer than the arc-disjoint ones because they need
to add an extra edge for each vertex in order to insure
vertex-disjointness as well as arc-disjointness.
Note, however, that, as with the set-disjoint case, the time
for vertex- and arc-disjoint triple generation again may well be growing
at slightly less than the cubic rate of our worst-case analysis.

Our overall running times also do not precisely track triple counts
The times for the path-disjoint versions of the instances are
indeed significantly slower than for the set-disjoint versions (except in the
case of R1400), but once again the arc-disjoint times tend to be faster than
those for the vertex-disjoint versions.
Further study will be needed to understand these discrepancies.

\subsection{Cost Reduction}

\begin{table}
\begin{center}
\begin{tabular}{c|c | r@{\ \ }|c @{\ \ } c  c @{\ \ }|c @{\ \ } c c @{\ \ } c| c c}
& & &\multicolumn{3}{c|}{Set-Disjoint} & \multicolumn{6}{c}{Path-Disjoint} \\ \hline
&& & & & & \multicolumn{2}{c}{Vertex-Disjoint} & \multicolumn{2}{c|}{Arc-Disjoint} & \multicolumn{2}{c}{Unbounded} \\ 
$|V|$ & Class & $|C|$ &\multicolumn{2}{c}{\greedy} & {\opt} & \multicolumn{2}{c}{\greedy} & \multicolumn{2}{c|}{\greedy} &  \multicolumn{2}{c} {\dsj} \\ \hline
& (C1,F1) & 984 & 16.2 & (83.8) & 16.2 & 15.1 & (7.1) & 15.1 & (0.0) & 10.7 & (29.2) \\
& (C2,F1) & 492 & 24.0 & (76.0) & 24.0 & 23.2 & (3.6) & 23.1 & (0.4) & 19.1 & (17.1) \\
& (C2,F2) & 492 & 25.7 & (74.3) & 25.7 & 23.9 & (6.9) & 23.9 & (0.2) & 19.0 & (20.6) \\
984 & (C4,F1) & 246 & 37.9 & (62.1) & 37.9 & 37.5 & (1.1) & 37.4 & (0.3) & 34.2 & ( 8.5) \\
& (C4,F4) & 246 & 40.1 & (59.9) & 40.1 & 38.4 & (4.4) & 38.3 & (0.2) & 33.3 & (13.0) \\
& (C8,F1) & 123 & 55.4 & (44.6) & 55.4 & 55.1 & (0.6) & 55.1 & (0.0) & 52.3 & ( 5.2) \\
& (C8,F8) & 123 & 58.2 & (41.8) & 58.2 & 57.2 & (1.8) & 57.1 & (0.1) & 51.2 & (10.3) \\
& & & & & & & & & &\\
& (C1,F1) &558 & 16.0 & (84.0) & 16.0 & 14.9 & ( 7.3) & 14.9 & ( 0.0) & 10.7 & (28.3) \\
& (C2,F1) & 279 & 24.6 & (75.4) & 24.6 & 23.6 & ( 4.2) & 23.5 & ( 0.2) & 19.6 & (16.7) \\
& (C2,F2) & 279 & 25.8 & (74.2) & 25.8 & 24.5 & ( 5.3) & 24.4 & ( 0.1) & 19.2 & (21.4) \\
558 & (C4,F1) & 140 & 37.8 & (62.2) & 37.8 & 37.1 & ( 1.7) & 37.1 & ( 0.0) & 34.4 & ( 7.3) \\
& (C4,F4) & 140 & 42.1 & (57.9) & 42.1 & 41.2 & ( 2.2) & 41.1 & ( 0.3) & 34.6 & (15.7)) \\
& (C8,F1) &  70 & 56.3 & (43.7) & 56.3 & 56.1 & ( 0.3) & 56.1 & ( 0.0) & 53.0 & ( 5.6) \\
& (C8,F8) &  70 & 60.0 & (40.0) & 60.0 & 59.1 & ( 1.4) & 58.7 & ( 0.7) & 54.7 & ( 6.8) \\
& & & & & & & & & &\\
& (C1,F1) & 300 & 24.0 & (76.0) & 24.0 & 23.0 & ( 4.2) & 23.0 & ( 0.0) & 14.9 & (35.5) \\
& (C2,F1) & 150 & 33.7 & (66.3) & 33.7 & 32.6 & ( 3.2) & 32.6 & ( 0.0) & 25.2 & (22.7) \\
& (C2,F2) & 150 & 36.1 & (63.9) & 36.1 & 34.8 & ( 3.5) & 34.8 & ( 0.0) & 24.7 & (29.1) \\
300 & (C4,F1) & 75 & 49.1 & (50.9) & 49.1 & 47.7 & ( 2.7) & 47.7 & ( 0.0) & 41.1 & (14.0) \\
& (C4,F4) & 75 & 50.8 & (49.2) & 50.8 & 50.0 & ( 1.6) & 49.9 & ( 0.3) & 40.5 & (18.7) \\
& (C8,F1) &  38 & 64.5 & (35.5) & 64.5 & 63.7 & ( 1.2) & 63.7 & ( 0.0) & 58.4 & ( 8.3) \\
& (C8,F8) &  38 & 68.4 & (31.6) & 68.4 & 68.4 & ( 0.0) & 68.4 & ( 0.0) & 57.4 & (16.2) \\
& & & & & & & & & &\\
& (C1,F1) & 190 & 28.7 & (71.3) & 28.7 & 27.7 & ( 3.5) & 27.7 & ( 0.0) & 21.0 & (24.3) \\
& (C2,F1) & 95 & 40.3 & (59.7) & 40.3 & 39.6 & ( 1.8) & 39.6 & ( 0.0) & 33.4 & (15.7) \\
& (C2,F2) & 95 & 41.7 & (58.3) & 41.7 & 40.2 & ( 3.5) & 40.1 & ( 0.3) & 33.2 & (17.3) \\
190 & (C4,F1) & 48 & 55.8 & (44.2) & 55.8 & 55.2 & ( 1.1) & 55.2 & ( 0.0) & 51.7 & ( 6.4) \\
& (C4,F4) & 48 & 55.4 & (44.6) & 55.4 & 54.4 & ( 1.9) & 54.4 & ( 0.0) & 47.7 & (12.3) \\
& (C8,F1) &  24 & 67.1 & (32.9) & 67.1 & 66.7 & ( 0.6) & 66.7 & ( 0.0) & 60.4 & ( 9.4) \\
& (C8,F8) &  24 & 70.8 & (29.2) & 70.8 & 70.8 & ( 0.0) & 70.8 & ( 0.0) & 61.2 & (13.5) \\
& & & & & & & & & &\\
& (C1,F1) & 100 & 24.1 & (75.9) & 24.1 & 23.6 & ( 2.1) & 23.6 & ( 0.0) & 14.3 & (39.4) \\
& (C2,F1) & 50 & 34.2 & (65.8) & 34.2 & 34.0 & ( 0.6) & 34.0 & ( 0.0) & 23.4 & (31.2) \\
& (C2,F2) & 50 & 35.4 & (64.6) & 35.4 & 34.4 & ( 2.8) & 34.4 & ( 0.0) & 22.6 & (34.3) \\
100 & (C4,F1) & 25 & 46.8 & (53.2) & 46.8 & 46.4 & ( 0.9) & 46.4 & ( 0.0) & 38.4 & (17.2) \\
& (C4,F4) & 25 & 49.2 & (50.8) & 49.2 & 48.4 & ( 1.6) & 48.4 & ( 0.0) & 37.2 & (23.1) \\
& (C8,F1) & 13 & 61.5 & (38.5) & 61.5 & 61.5 & ( 0.0) & 61.5 & ( 0.0) & 48.5 & (21.3) \\
& (C8,F8) & 13 & 63.8 & (36.2) & 63.8 & 62.3 & ( 2.4) & 61.5 & ( 1.2) & 47.7 & (22.5) \\ \hline
\end{tabular}
\caption{Average percentages of customer vertices in algorithmically
generated covers for all seven classes of our $|V| \in \{984, 558, 300, 190, 100\}$ synthetic
instances.
Entries in parentheses are the percentage reduction in cover size
from the results in the preceding column.
}\label{coverreduction}
\end{center}
\vspace{-.25in}
\end{table}

In this section we consider the savings our heuristics can provide
over simply choosing the cover consisting
of the set $C$ of all customer vertices, which is always a feasible solution,
and would be the default solution for network path monitoring if we did not use our
proposed monitoring scheme.

We first consider the results for our synthetic instances.
Table \ref{coverreduction} presents our results for the seven classes of
instances for each value of $|V| \in \{984, 558, 300, 190, 100\}$.
(For space reasons, we omit the results for
$|V| \in \{250, 220, 50\}$, but they do not evince any strikingly different
behavior.)
For each algorithm and instance class,
we give both the average percentage of customers that are in
the covers produced by the algorithm, and (in parentheses) the percentage
reduction over the results in the column to the left.
Here ``{\greedy}'' stands for {\gr4}, the algorithm that takes the best of
400 independent runs of {\greedy}.
``Opt'' stands for the cover produced by our {\lb} code for set-disjoint instances,
which for these instances was always optimal.
``DP'' stands for the optimum when we do not require our
vertex-disjoint paths to be shortest paths, as computed
by the linear-time dynamic programming algorithm described in Section \ref{sec:trees}.
This value is not relevant to either of our applications, but does provide a
lower bound on what is possible in the path-disjoint case.

These results clearly show that using our scheme for
disjoint path monitoring offers a substantial
reduction in the number of vertices needed to cover all the customers.
In the set-disjoint case, the reduction is typically by a factor of 4 or more
for the (C1,F1) instances, in which all vertices are customers, and grows to a
factor of 6 when $|V| \in \{558, 984\}$.
As the proportion of vertices that are customers is reduced for a given value
of $|V|$, the savings
decreases, but is always more than 29\%.
Perhaps not surprisingly, when not all vertices are customers, the reduction
is greater in the (C$k$,F1) classes, where all vertices are potential facility
locations, than in the (C$k$,F$k$) classes, where the only potential facility
locations are the customers themselves.
Note that, as remarked before, {\gr4} finds optimal set-disjoint solutions
in all cases, so we omit
the parenthetical ``percentage reduction'' column for {\opt}.

Going from the set-disjoint to the path-disjoint instances requiring
vertex-disjointness typically offers a significant reduction (up to 7\% in
the (C1,F1) classes with $|V| \in \{558, 984\}$),
although the percentage reduction declines with $|V|$ and $|C|$.
Only requiring arc-disjoint paths often yields a small additional improvement,
although typically less than 1\%.
Allowing unbounded vertex-disjoint paths would allow a much more significant
reduction, up to almost 40\% in one case, and typically more for the
(C$k$,F$k$) classes, $k \geq 2$, than in the corresponding (C$k$,F1) classes,
although the actual sizes of the covers still remain lower in the latter case.

\begin{table}
\begin{center}
\begin{tabular}{c | c|r @{\ \ } c r @{\ \ } c|r @{\ \ } c r @{\ \ } c| c c}
& &\multicolumn{4}{c|}{Set-Disjoint} & \multicolumn{6}{c}{Path-Disjoint} \\ \hline
& Approx & & & & & \multicolumn{2}{c}{Vertex-Disjoint} & \multicolumn{2}{c|}{Arc-Disjoint} & \multicolumn{2}{c}{Unbounded} \\ 
Instance & $|C|$ &\multicolumn{2}{c}{\greedy} & \multicolumn{2}{c|}{\opt} & \multicolumn{2}{c}{\greedy} & \multicolumn{2}{c|}{\greedy} &  \multicolumn{2}{c} {\dsj} \\ \hline
R5500 & 3200 & 37.7 & (62.3) & 17.0 & (55.0) &  2.8 & (83.5) &  2.8 & ( 0.0) &  0.6 & (80.0)\\
R3000 & 1500 &  1.8 & (98.2) &  1.8 & ( 0.0) &  1.3 & (25.9) &  1.3 & ( 0.0) &  0.1 & (90.0)\\
R1400 &  900 & 14.3 & (85.7) & 14.3 & ( 0.0) &  3.2 & (77.6) &  3.2 & ( 0.0) &  0.3 & (90.0)\\
R1000 &  900 & 35.5 & (64.5) & 10.1 & (71.5) &  3.8 & (62.0) &  3.8 & ( 0.0) &  0.9 & (77.1)\\
 R700 &  500 &  4.9 & (95.1) &  4.9 & ( 0.0) &  2.8 & (43.5) &  2.8 & ( 0.0) &  0.6 & (76.9)\\
 R500 &  500 & 13.9 & (86.1) & 13.9 & ( 0.0) & 12.4 & (10.8) & 12.4 & ( 0.0) &  1.9 & (84.8)\\
 R300 &  200 &  4.1 & (95.9) &  4.1 & ( 0.0) &  3.6 & (11.1) &  3.6 & ( 0.0) &  2.7 & (25.0)\\
 R200 &  200 &  7.6 & (92.4) &  7.6 & ( 0.0) &  6.7 & (11.8) &  6.7 & ( 0.0) &  2.2 & (66.7)\\
R100a &  100 & 10.9 & (89.1) & 10.9 & ( 0.0) &  7.8 & (28.6) &  7.8 & ( 0.0) &  4.7 & (40.0)\\
R100b &  100 & 10.9 & (89.1) & 10.9 & ( 0.0) &  7.8 & (28.6) &  7.8 & ( 0.0) &  4.7 & (40.0)\\ \hline
\end{tabular}
\caption{Average percentages of customer vertices in algorithmically
generated covers for our 10 ISP-based instances.
Entries in parentheses are the percentage reduction in cover size
from the results in the preceding column.
}\label{Rcoverreduction}
\end{center}
\vspace{-.25in}
\end{table}

The results for our ISP-based instances are presented in Table \ref{Rcoverreduction},
where we only include the (C1,F1) versions of those instances where we did not
know the types of routers and so could not deduce what the actual sets $F$ and $C$
should be.
The resulting value $|C|/|V| = 1$ is reasonably close to the corresponding
values for the instances where
we {\em can} estimate the actual sets $F$ and $C$; for these the ratio ranges from
about 0.6 to 0.9.
Because the instances where we {\em do} have this information are
proprietary, the column labeled by ``$|C|$'' only lists the value
of $C$ rounded to the nearest 100.
The average cover sizes reported, however, are given as percentages of
the true value of $|C|$.
Note that in this table we do include a parenthetical improvement column
for {\opt} over {\gr4}, since here there {\em are} improvements for two of
the instances.

As we can see, the effectiveness of our disjoint path monitoring scheme
is even more substantial for these instances than for our synthetic ones.
The best point of comparison is to the (C1,F1) variants of the latter, where the
smallest ratio of set-disjoint cover size to $|C|$ was roughly 16\%.
Here the {\em largest} percentage is 17\%, and four of the ten are less
than 10\%, with the percentage for R3000 being just 1.8\%.
Moreover, going to the path-disjoint case now offers major improvements,
in comparison to the at-most 7.3\% improvement for our synthetic instances.
Here the range is from 10.8\% to 83.5\%, and now all but one of our instances
have covers with size less than 8\% of $|C|$.

Allowing arc-disjoint paths offers no further improvement.
However, if one could consider arbitrarily long vertex-disjoint paths, even more
substantial reductions are possible.
Now half the instances have covers with size less than 1\% of $|C|$.
This should not be a surprise, however.
As we already observed in Section \ref{propertysect},
typically 95\% or more of the vertices in our ISP-based
instances are in a single large 2-connected component.
By the definition of 2-connectedness, {\em any} two vertices
in such a component will constitute a cover.
Additional vertices are only needed to handle the few
additional small components hanging off the large one.

\section{Further Results and Directions}\label{section:further}
The problems described here can be generalized in a variety of
directions, for two of which we already have some preliminary results.

\begin{enumerate}

\item Variants of SDFL and PDFL in which we omit the requirement that
every customer vertex also be a potential facility location (i.e., that $C \subseteq F$).
In our monitoring application, it might well be the case that some of
the customer vertices are housed in locations that do not have enough
free physical space to house our monitoring equipment, or are in leased
space into which we do not have permission to add additional equipment.

\item Variants in which for each vertex a given nonnegative cost must be paid if that
vertex is used as a facility, and our goal is to minimize total cost.
Note that the second variant includes the first as a special case, since
we can simply set the cost of each member of $F$ to be 1 and the cost of all
other vertices to be $|F|+1$.

\end{enumerate}

Most of our algorithms and lower bounds carry over more-or-less directly
to both types of generalization, with the two exceptions being
the linear-time algorithms of Section \ref{treetheo} and Section \ref{sec:pdlb}
for computing the optimum covers in the special case of trees
and our general lower bound for the path-disjoint case.
Linear-time algorithms still exist for both of these tasks, but the simple combinatorial
characterizations that previously sufficed must now be replaced by dynamic programs
of some complexity.
Our other algorithms and lower bounds carry over with almost no change for the
first type of generalization.
For the second type, the integer programs for computing the optimum and for computing
the Hitting Set lower bound must change their objective functions so that the
variables $x_f$ are multiplied by the corresponding costs.
And the various cover-constructing heuristics must modify their greedy choices
by taking costs into account, for instance by choosing the $f \in F$ which yields
the highest ratio (increase in coverage)/(cost of $f$).

One question that arises for the first type of variant that was not an issue
in our original version of the problems is the question of feasibility.
In the original version, where $C \subseteq F$, there was always a feasible
cover, mainly the set $C$ itself, with each customer covering itself.
When not all customers are potential facility locations, it may be that {\em no}
cover exists.
Thus the first question one must ask is whether such a cover exists.
This will be true if and only if the full
set $F$ is a valid cover, a property that can easily be checked
in $O(n^2m)$ tine by generating
the corresponding SCP problem and verifying that there is a valid triple
for each customer.

We leave the empirical testing of our algorithms for these generalizations of
SDFL and PDFL for further research.
At present the variety of potential instance classes is too large for any
one or two choices to provide meaningful insights, although future research
into potential applications may identify particular instance classes whose
practical relevance justifies such study.

\medskip
\noindent
{\bf {\Large Acknowledgment}}

\smallskip
\noindent
David S. Johnson, the lead author of this paper, passed away on 
March 8, 2016.
David's co-authors dedicate this paper to his memory.
The authors thank David Applegate for helpful discussions and in particular
for simplifying the statement and proof of Theorem \ref{treetheo}.
The authors also thank Rodrigo Toso for implementing the version of 
\genetic\ used in the experiments.

 

\bibliographystyle{plainnat}
\bibliography{journal}

\end{document}